\documentclass[sn-mathphys-ay,pdflatex]{sn-jnl} 
\providecommand{\doiurl}[1]{\url{#1}}

\usepackage{graphicx}
\usepackage{multirow}
\usepackage{amsmath,amssymb,amsfonts}
\usepackage{amsthm}
\usepackage{mathrsfs}
\usepackage[title]{appendix}
\usepackage{textcomp}
\usepackage{manyfoot}
\usepackage{booktabs}
\usepackage{algorithm}
\usepackage{algorithmicx}
\usepackage{algpseudocode}
\usepackage{listings}
\usepackage{bm}
\usepackage{comment}
\usepackage{wrapfig}
\usepackage{lineno}
\usepackage{xcolor}
\usepackage{svg}
\usepackage{subfigure}
\usepackage{subcaption}
\usepackage{soul}
\usepackage{ulem} 
\usepackage{subcaption}  
\usepackage{appendix}
\usepackage{orcidlink}
\usepackage{tikz}
\usetikzlibrary{positioning, arrows.meta, calc}

\numberwithin{equation}{section}

\theoremstyle{thmstyleone}
\newtheorem{theorem}{Theorem}

\theoremstyle{thmstyletwo}

\newtheorem{remark}{Remark}

\theoremstyle{thmstylethree}

\newtheorem{lemma}{Lemma}
\title[Article Title]{Integrating Household Dynamics in Stochastic Epidemic Modeling: An SDE Approach to the SIR Framework}

\author[1]{\fnm{Houda} \sur{Yaqine}}\email{houda.yaqine@uni-bielefeld.de}

\author*[1,2]{\fnm{Christiane} \sur{Fuchs}}\email{christiane.fuchs@uni-bielefeld.de}

\affil*[1]{\orgdiv{Faculty of Business Administration and Economics}, \orgname{Bielefeld University},  \state{Bielefeld}, \country{Germany}}

\affil[2]{\orgdiv{Computational Health Center}, \orgname{Institute of Computational Biology, Helmholtz Zentrum München},  \city{Neuherberg}, \country{Germany}}

\begin{document}

\abstract{Understanding infectious disease spread remains a critical public health challenge, particularly given the interplay between household dynamics and community transmission patterns. Traditional epidemiological models often oversimplify these dynamics by treating populations as homogeneous, failing to capture crucial household-level interactions that can significantly impact disease spread. 
This paper introduces a new stochastic differential equation model extending the SIR framework by capturing the randomness in disease spread and incorporating household structure and heterogeneous mixing patterns. The model divides the population into groups based on age and household size, includes subpopulation-targeted lockdown parameters and constructs detailed contact matrices accounting for both public and within-household interactions.

Through the approximation of Markov jump processes by branching processes near the disease free equilibrium, we derive the basic reproduction number of our model and conduct global sensitivity analysis using Sobol indices to identify influential factors. Our simulations reveal that incorporating household structure leads to substantially different predictions compared to traditional models, particularly in epidemic timing and peak intensity. The stochastic framework captures important variations in outbreak trajectories overlooked by deterministic approaches, especially during early and peak phases. This work contributes to both mathematical epidemiology and practical public health planning by providing a sophisticated mathematical understanding of how population structure and randomness influence disease dynamics, offering insights for intervention strategies where household transmission plays a significant role.}


\maketitle
%\linenumbers

\section{Introduction}\label{sec1}
Infectious diseases pose a significant and persistent challenge to global public health, leading to considerable mortality and disability worldwide \citep{naghavi2024global, ledesma2024global}. Understanding the dynamics of disease transmission is crucial to developing effective control measures and predicting the impact of intervention policies.  To address this complex challenge, researchers have turned to mathematical modeling \citep{grassly2008mathematical} as a powerful analytical framework that can capture and predict disease transmission patterns. While various approaches exist, including network-based models and agent-based simulations, compartmental models have emerged as particularly valuable tools due to their balance of simplicity, interpretability, and predictive power. Among these compartmental approaches, one of the most widely used models to study the spread of infectious diseases is the SIR model \citep{kermack1927contribution}. This model divides the population into three compartments: Susceptible (those who can contract the disease), Infected (those who have the disease and can spread it), and Recovered (those who have recovered and are immune). Originally designed for populations with uniform mixing patterns, the SIR model has evolved to include variations such as SIS (Susceptible-Infected-Susceptible), SEIR (Susceptible-Exposed-Infected-Recovered), age-structured SIR, and spatial SIR models to account for more complex disease transmission dynamics.

Beyond these structural variations, the mathematical implementation of the SIR model can take different forms, each with distinct advantages for specific modeling scenarios. These include deterministic formulations such as ordinary differential equations (ODEs) and stochastic formulations like stochastic differential equations (SDEs) \citep{Allen2010,Brauer2019}. SDEs are particularly useful as they account for random variations in disease spread due to factors such as individual behavior, fluctuations in contact rates, and environmental influences. This stochastic approach is especially relevant for diseases like SARS-CoV-2, where transmission is highly stochastic and driven by super-spreading events \citep{lau2020characterizing} .

SDE models overcome several limitations of deterministic ODE approaches in epidemiology. For identical initial conditions, ODEs predict identical outcomes, missing the inherent randomness seen in actual disease outbreaks \citep{allen2008introduction}. This becomes problematic with small infection numbers, where chance events significantly shape outbreak patterns \citep{allen2017primer}. Studies found that ODEs overestimate disease persistence during local extinctions, while SDEs capture both persistence and extinction scenarios \citep{Roberts2014} more appropriately. Additionally, unlike ODEs with their state-continuous population changes, real disease spread occurs in discrete jumps affected by demographic and environmental factors \citep{Raczynski1996}. SDEs can better approximate these discrete jumps through random fluctuations while maintaining mathematical tractability, offering a middle ground between simpler ODE models and more complex agent-based methods \citep{Niemann2021}. 

Traditional compartmental models, regardless of whether they use ODE or SDE formulations, assume homogeneous mixing of populations within each compartment. In the standard three-compartment SIR model, individuals from the entire population meet randomly without considering the influence of repeated encounters. In reality, people who live in the same household are in contact with each other more frequently than with those outside their household, resulting in higher transmission rates within households. The importance of household structure in disease transmission has been well-documented: studies by \cite{hilton2019incorporating} and \cite{liu2021modelling} emphasize the role of household structure in disease dynamics, whereas \cite{Chisholm2020} and \cite{goeyvaerts2018household} underscore the need for models that account for complex household structures and non-random contacts. \cite{Meszaros2020} and \cite{Black2022} provide specific examples of how household transmission shapes epidemic outcomes. Although stochastic household epidemic models have been developed using continuous-time Markov chains, and deterministic frameworks exist for household-structured models with two-level mixing \citep{Bayham2016}, a comprehensive stochastic framework using SDEs that simultaneously tracks both individual health states and household-level dynamics remains underdeveloped. Such a dual-level SDE approach would capture the interplay between within-household and between-household transmission while accounting for demographic stochasticity. It would provide a tractable middle ground between computationally intensive Markov jump process simulations and deterministic approximations. 

To address this gap, we develop a dual system of coupled SDEs for the SIR model, one tracking individual health states and another tracking household health states. Specifically, we extend the deterministic framework of \cite{Bayham2016} in two key ways: (1) by systematically deriving SDEs from the approximation of Markov jump processes (MJP) by diffusion processes \citep{Fuchs2013}, providing a tractable way to represent stochastic fluctuations in population interactions; and (2) by constructing explicit contact matrices that capture heterogeneous mixing patterns based on age and household size. While Bayham's model uses aggregate transmission rates, our contact matrices reflect social behavior and mixing preferences specific to different subpopulations. This dual-level stochastic framework with structured contact patterns contributes a new perspective to the field of epidemic modeling.

This paper is structured as follows: Section~\ref{sec: model} presents our model, detailing its construction process and the development of the contact matrices. Section~\ref{sec:branching_process_approx} focuses on the derivation of the basic reproduction number to gain insight about the model's threshold behavior. In Section~\ref{sec:simulations}, we conduct several numerical analyses: we simulate the model to investigate the impact of public policy on exposure intensity, perform a global sensitivity analysis of the threshold parameter, and compare our model's deterministic version with other multitype SIR models to highlight key structural differences. Section~\ref{sec:conclusion} concludes with a discussion of our findings. We address the implications of our work and identify key points for future research that could further enhance this contribution to the field of epidemic modeling.

\section{Model Formulation}\label{sec: model}

Our model carefully considers how disease states evolve at both individual and household levels. At the population level, we model the numbers of individuals in susceptible~($\bm S$), infected ($\bm I$), and recovered ($\bm R$) states. Similarly, for households, we model their health status using states $\bm H_S$ (susceptible household), $\bm H_I$ (infected household), and~$\bm H_R$ (recovered household).  
Since a household comprises multiple individuals, we characterize a household's susceptibility (i.\,e., being in state~$\bm H_S$) as the complete susceptibility of all of its members; a household's infection event (i.\,e., moving to state~$\bm H_I$) as the initial infection occurring within an entirely susceptible household. 
For simplicity, we assume that infected individuals inevitably infect all other members of their households. A household's recovery event (i.,e., moving to state~$\bm H_R$) takes place with the final recovery of all household members. This simplification keeps the model manageable by avoiding pathways from $\bm H_I$ back to a state of non-infected but partially susceptible households.

The transitions between these states, whether for individuals or households,  follow Markov jump processes, which we later approximate by diffusion processes described by SDEs \citep{Fuchs2013}. The resulting drift terms coincide with the deterministic household-structured model of \cite{Bayham2016}, whose notation we adopt for consistency. However, unlike adding noise to existing ODEs, our approach derives both drift and diffusion components from the transition structure of the MJP, ensuring that stochastic fluctuations are properly scaled to the process dynamics.

\begin{figure}[H]
    \centering
    \includegraphics[width=0.9\linewidth]{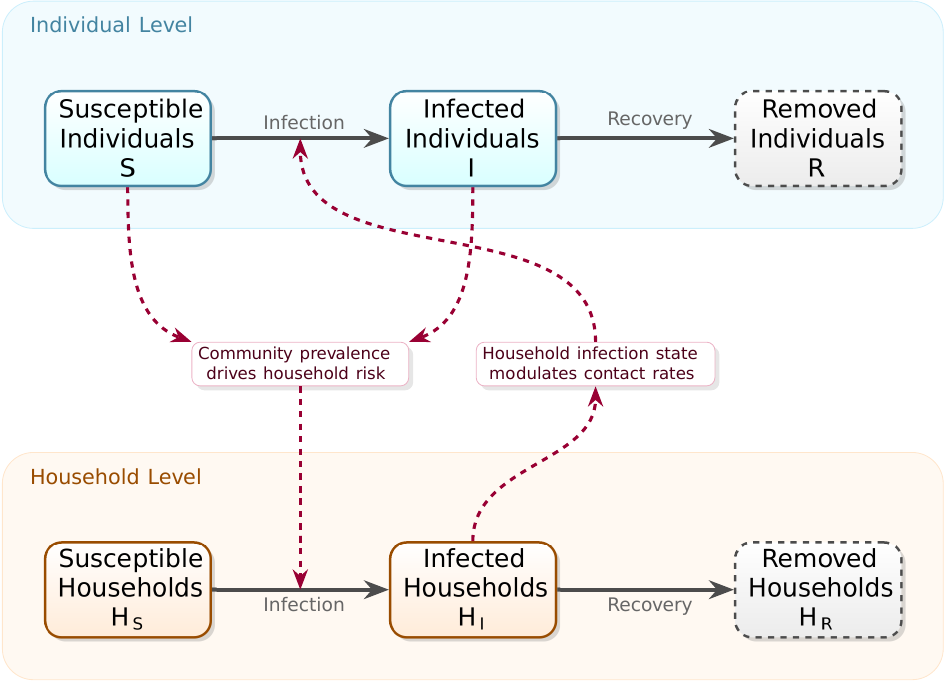}
    \caption{Model structure with two types of connections. Solid arrows represent the physical state transitions (infection and recovery). Dashed lines depict the interaction between scales: individual prevalence drives household risk, which in turn reshapes individual transmission dynamics.}
    \label{fig:placeholder}
\end{figure}

The key challenge in our model lies in the distinct nature of individual and household transitions. While both levels follow Markov jump processes, combining them into a single process would disproportionately increase the complexity of the model. 
We initially separate them because household state changes $(\bm H_S \xrightarrow{} \bm H_I \xrightarrow{} \bm H_R)$, as defined by the collective status of all members,  do not align directly with individual transitions $(\bm S \xrightarrow{} \bm I \xrightarrow{}\bm R)$. For instance, when a susceptible individual ($\bm S$) becomes infected ($\bm I$), their household will simply remain in the same state $\bm H_I$ if it already contains an infected person. As a result, we construct separate systems of SDEs for individuals $(\bm S, \bm I, \bm R)$ and households $(\bm H_S, \bm H_I, \bm H_R)$. These systems, while distinct, remain coupled to capture how infection dynamics at each level influence the other.

\subsection{Assumptions and Notation}

Infection events of individuals require infectious contacts, and recovery events presuppose a healing process in the individual's body. The probabilities for such events are modelled to be proportional to contact and recovery rates, and we assume that these rates depend on the individuals' age and the size of the household they are living in.
We thus divide the population into~$K=A\times L$ subpopulations defined by combinations of $A$~age categories and $L$~household sizes, with combination~$k=(l-1)A+a$ representing the $a$-th age category and $l$-th household size. 

Contact rates between individuals differ depending on whether contacts occur in public locations or within households.  Let $\bm S$, $\bm I$, and $\bm R$ be $K$-dimensional column vectors representing  numbers of susceptible, infected, and recovered individuals in each subpopulation, respectively. Similarly, let $\bm H_S$, $\bm H_I$, and $\bm H_R$ be $L$-dimensional column vectors representing the numbers of susceptible, infected, and recovered households in each size. The population is considered closed, with no births, deaths, or migration during the modeled time window. Thus, sizes of recovered subpopulations can be deduced from the other states, and the vectors $\bm{X} = (S_1,I_1,\dots, S_K, I_K)^T$ and $\bm{U} = (H_{S,1},H_{I,1},\dots, H_{S,L},H_{I,L})^T$ provide complete information regarding state sizes for individuals and households, where the subscript~$k$ of~$S_k$ and~$I_k$ refers to the combination of age category and household size as defined above. We further introduce 
$\bm{Y} = 
(\bm{X}^T,\bm{U}^T)^T$.
Let $\bm N=(N_1,\dots,N_K)^T$ indicate the numbers of individuals in each subpopulation such that $\sum_{k=1}^K N_k=N$ is the total size of the population, and $\bm N_H=(N_{H,1},\dots,N_{H,K})^T$ where $N_{H,k}$ represents the number of households (of the size corresponding to subpopulation $k$) in which individuals from subpopulation $k$ reside. Since households typically contain individuals of multiple ages, the same physical household may be counted in multiple entries of $\bm N_H$. $\bm H=(H_{1},\dots,H_{L})^T$ stands for the numbers of households in each size class. 

Next, we introduce a $K\times L$ conformability matrix~$\bm{Z}$ of zeros and ones where each row $k$ represents a specific subpopulation (referring to the combination of age category and household size as above), and each column $l$ stands for a household size category. An entry of one in position~$(k,l)$ indicates that individuals from subpopulation~$k$ live in households of size category~$l$.
This allows us to identify which subpopulations belong to households of a particular size. 
An example of a conformability matrix $\bm{Z}$ is given in Section~\ref{conform}.

We further consider $K\times K$ contact matrices $\bm C^p$ and $\bm C^h$ that represent the rates of contact between individuals from the~$K$ subpopulations 
in the public and household settings, respectively. 
Section~\ref{section_contact} gives more details about the construction of these matrices.
Table~\ref{tab: nota} provides an overview of all variables and parameters used throughout this paper, along with their descriptions for easy reference.
\begin{table}[htbp]
\caption{Model variables, parameters and contact rate components for individual-household disease transmission dynamics. The subscript $+$~denotes a positive domain ($>0$), $0$ a non-negative one ($\geq 0$).}
\begin{tabular}{p{0.355\textwidth}|p{0.5\textwidth}|p{0.1\textwidth}}
\hline
\textbf{Symbol} & \textbf{Description} & \textbf{Domain} \\
\hline\hline
$K$ & Total number of subpopulations ($K = A \times L$) & $\mathbb{N}_+$ \\
\hline
$A$ & Number of age categories & $\mathbb{N}_+$ \\
\hline
$L$ & Number of household size categories & $\mathbb{N}_+$ \\
\hline
\multicolumn{3}{l}{\textit{Population state variables:}} \\
\hline
$\bm{S}, \bm{I}, \bm{R}$ &  Numbers of susceptible, infected, and recovered individuals & $\mathbb{N}_0^K$ \\
\hline
$\bm{H_S}, \bm{H_I}, \bm{H_R}$ &  Numbers of susceptible, infected, and recovered households & $\mathbb{N}_0^L$ \\
\hline
$\bm{X} = (S_1,I_1,\dots, S_K, I_K)^T$ & State variables representing susceptible and infected individuals & $\mathbb{N}_0^{2K}$ \\
\hline
$\bm{U} = (H_{S,1},H_{I,1},\dots, H_{S,L},H_{I,L})^T$ & State variables representing susceptible and infected households & $\mathbb{N}_0^{2L}$ \\
\hline
$\bm{x} = (s_1,i_1\dots,s_K,i_K)$ & Normalized state variables for individuals (each component divided by the corresponding subpopulation size $N_k$) & $[0,1]^{2K}$ \\
\hline
$\bm{h} = (h_{S,1},h_{I,1},\dots,h_{S,L},h_{I,L})$ & Normalized state variables for households (each component divided by the corresponding household class size $H_l$)& $[0,1]^{2L}$ \\
\hline
$\bm{Y}=(\bm{X}^T,\bm{U}^T)^T$

& Combined state variables of susceptible and infected individuals and households & $\mathbb{N}_0^{2K+2L}$ \\
\hline
\multicolumn{3}{l}{\textit{Population size variables:}} \\
\hline
$\bm{N} = (N_1,\dots,N_K)$ & Numbers of individuals in each subpopulation & $\mathbb{N}_+^K$ \\
\hline
$N = \sum_{k=1}^K N_k$ & Total population size & $\mathbb{N}_+$ \\
\hline
$\bm{N_H} = (N_{H,1},\dots,N_{H,K})$ & Numbers of households containing individuals from the respective subpopulations & $\mathbb{N}_+^K$ \\
\hline
$\bm{H} = (H_{1},\dots,H_{L})$ & Numbers of households in the each size class & $\mathbb{N}_+^L$ \\
\hline
$\bm{Z}=(z_{kl})_{k,l}$ & Conformability matrix & $\{0,1\}^{K \times L}$ \\
\hline
\multicolumn{3}{l}{\textit{Contact matrices and their components:}} \\
\hline
$\bm{C^p}, \bm{C^h}$ & Contact matrices for public and household settings & $\mathbb{R}_+^{K \times K}$ \\
\hline
$c_{ij}^p, c_{ij}^h$ & Components of $\bm{C^p}$ and $\bm{C^h}$: contact rates between subpopulations $i$ and $j$ in public and households & $\mathbb{R}_+$ \\
\hline
$\epsilon_k^p, \epsilon_k^h$ & Proportion of social contacts within own subpopulation $k$ in public/household settings& $[0,1]$ \\
\hline
$f_k^p, f_k^h$ & Interaction extent of subpopulation $k$ in public/household settings& $[0,1]$ \\
\hline
$a_k^p, a_k^h$ & Activity levels of subpopulation $k$ in public/household settings per day& $\mathbb{R}_+$ \\
\hline
$\omega_k$ & Public policy rate of subpopulation $k$ (from complete lockdown (0) to no restrictions (1)) & $[0,1]$ \\

\hline
\multicolumn{3}{l}{\textit{Epidemic parameters:}} \\
\hline
$\alpha$ & Transmission rate & $\mathbb{R}_+$ \\
\hline
$\beta_k$ & Recovery rate for subpopulation $k$ & $\mathbb{R}_+$ \\
\hline
$\nu_l$ & Recovery rate for households of size $l$ & $\mathbb{R}_+$ \\
\hline
\multicolumn{3}{l}{\textit{Transition probabilities:}} \\
\hline
$p_k$ & Main component of probability of an infection of a susceptible individual from subpopulation $k$ & $[0,1]$ \\
\hline
$p'_l$ & Main component of probability of an infection of a susceptible household of size category $l$ & $[0,1]$ \\
\hline
$\pi_k$ & Main component of probability of a recovery of an infective individual from subpopulation $k$ & $[0,1]$ \\
\hline
$\pi'_l$ & Main component of probability of a recovery of an infective household in size category $l$ & $[0,1]$ \\
\hline
\multicolumn{3}{l}{\textit{Stochastic process components:}} \\
\hline
$B_{k,1}(t)$, $B_{k,2}(t)$ & Brownian motions in SDEs for subpopulation $k$& $\mathbb{R}$ \\
\hline 
$B'_{l,1}(t)$, $B'_{l,2}(t)$ & Brownian motions in SDEs in household category $l$& $\mathbb{R}$ \\
\hline
\end{tabular}\label{tab: nota}
\end{table}
\newpage
\subsection{Contact Matrices}
\label{section_contact}

The definition of 'contact' in epidemiology is ambiguous and can refer to many aspects, which implies that there is no universal way of formalizing it. Instead, its expression varies according to how it is understood and the purpose of modeling. \cite{diekmann2000mathematical} identified two distinct interpretations: first, as the event of disease transmission itself, e.\,g.\  when infection is assumed to occur instantaneously upon contact; and second, as the pairing of two individuals for a time period during which multiple transmission events can occur. Our approach uses elements of the second interpretation by focusing on physical interactions of sufficient duration 
that provide opportunity for transmission, though we do not explicitly model varying contact durations.

In a large population, even in a vast city, an individual's interaction tends to be confined to a relatively small, repeated circle of contacts over time. This limitation influences the probability of meeting new people and spreading infection. Our framework represents this fact through the concept of contact rates, collected in a contact matrix which quantifies the frequency of potentially infectious interactions between different subpopulations. When modeling contact patterns, we account for heterogeneous interaction intensity between different subpopulations, reflecting variations in age, social activities, and workplace environments. Building on the same idea, \cite{hill2023implications} provided a framework for constructing matrices that represent different social structures and behaviors, from highly sociable to isolated subpopulations. We extend their approach by incorporating public policy parameters to explore intervention strategies. To that end, we further specify the components~$c_{ij}^p$ and~$c_{ij}^h$ of the contact matrices $\bm C^p$ and $\bm C^h$. For~$i,j=1,\ldots,K$, these represent the rates of contact between individuals from subpopulations~$i$ and~$j$ in the public and within households, respectively.

\subsubsection{Contacts in the public}

To capture both intra-group (within-subpopulation) and inter-group (between-subpopulations) interactions in public settings, we decompose the contact rate \( c^p_{ij} \) into

\begin{equation}\label{matrixconstr}
    c^p_{ij} = \left[ \underbrace{\delta_{ij} \epsilon^p_i}_{\text{self-interaction}} + \underbrace{(1 - \epsilon^p_i) f^p_j}_{\text{cross-group interaction}} \right] \underbrace{\omega_i \omega_j}_{\text{policy effect}}\!\!\!\!.
\end{equation}
This is motivated by the following model assumptions:
\begin{enumerate}
    \item \textbf{Self-interaction preference}: Individuals in subpopulation \( i \) allocate a proportion \( \epsilon^p_i \in [0,1] \) of their public contacts to members of their own subpopulation. This is modeled by the term \( \delta_{ij} \epsilon^p_i \), where the Kronecker delta \( \delta_{ij} \in \{0,1\} \) (which equals one if \( i = j \) and zero otherwise) ensures that this part applies only to within-group interactions. \\ 
    
    \item \textbf{Cross-group interaction}: The remaining proportion \( 1 - \epsilon^p_i \) of public contacts is distributed across other subpopulations. The number of such interactions depends on properties of the target subpopulation~$j$, namely on their public activity level~\( a^p_j \geq 0 \) (contacts per person from subpopulation~$j$ per unit time), population size \( N_j\), and self-interaction preference \( \epsilon^p_j \). The amount of contacts is captured by~\( f^p_j \in [0,1] \), a normalized term ensuring cross-group contacts are proportional to the 'availability' of subpopulation \( j \) in public spaces:  
    \[
    f^p_j = \frac{(1 - \epsilon^p_j) a^p_j N_j}{\sum_{k=1}^K (1 - \epsilon^p_k) a^p_k N_k}.
    \]  
    Here, \( 1 - \epsilon^p_j \) reflects subpopulation \( j \)'s openness to cross-group interactions in public, while \( a^p_j N_j \) gives the total number of public contacts generated by subpopulation $j$.  

    \item \textbf{Public policy restrictions}: Intervention measures such as lockdowns are incorporated via the factor~\( \omega_i \omega_j \), where \( \omega_i \in [0,1] \) represents the restriction level for subpopulation \( i \), with $\omega_i= 0$ indicating complete lockdown and $\omega_i= 1$ indicating no restrictions. The multiplicative form ensures that interactions cease if either subpopulation \( i \) or \( j \) is fully restricted (\( \omega_i = 0 \) or~\( \omega_j = 0 \)).   
\end{enumerate}
This structure of the contact rate~\eqref{matrixconstr} allows us to model both natural mixing preferences (via \( \epsilon^p_i \) and \( f^p_j \)) and targeted or uniform policy effects (via \( \omega_i \)). The resulting public contact matrix \( \mathbf{C}^p = (c^p_{ij})_{i,j=1,\ldots,K} \) reflects how subpopulations interact under specific behavioral and policy conditions. Figure \ref{fig:Contact_matrices}(a) displays an example of such a public contact matrix. The influence of~\( \epsilon^p_i \) and~\( \omega_i \) can be explored interactively using our \href{https://m8432n-houda-yaqine.shinyapps.io/shinyapp/}{Shiny app}. 

\begin{figure}[H]
  \centering
  \vspace{-0.5cm}
  \begin{minipage}{0.49\linewidth}
    \centering
    \includegraphics[width=\linewidth, trim=0cm 1cm 1cm 1cm, clip]{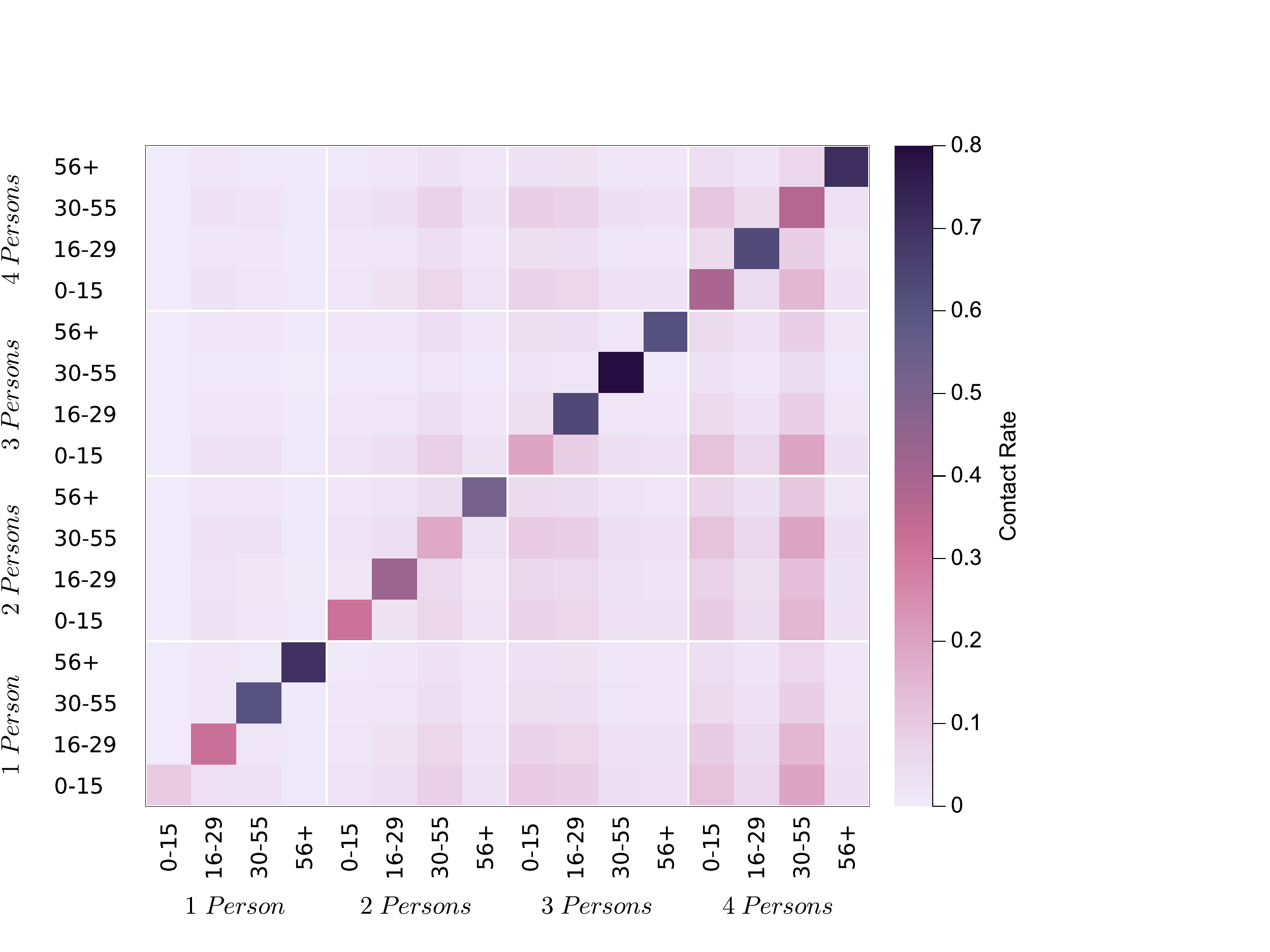}
    \vspace{-0.2cm} 
    \par\center\normalsize (a) in the public 
    \label{fig:cp}
  \end{minipage}
  \begin{minipage}{0.49\linewidth}
    \centering
    \includegraphics[width=\linewidth, trim=0cm 1cm 1cm 1cm, clip]{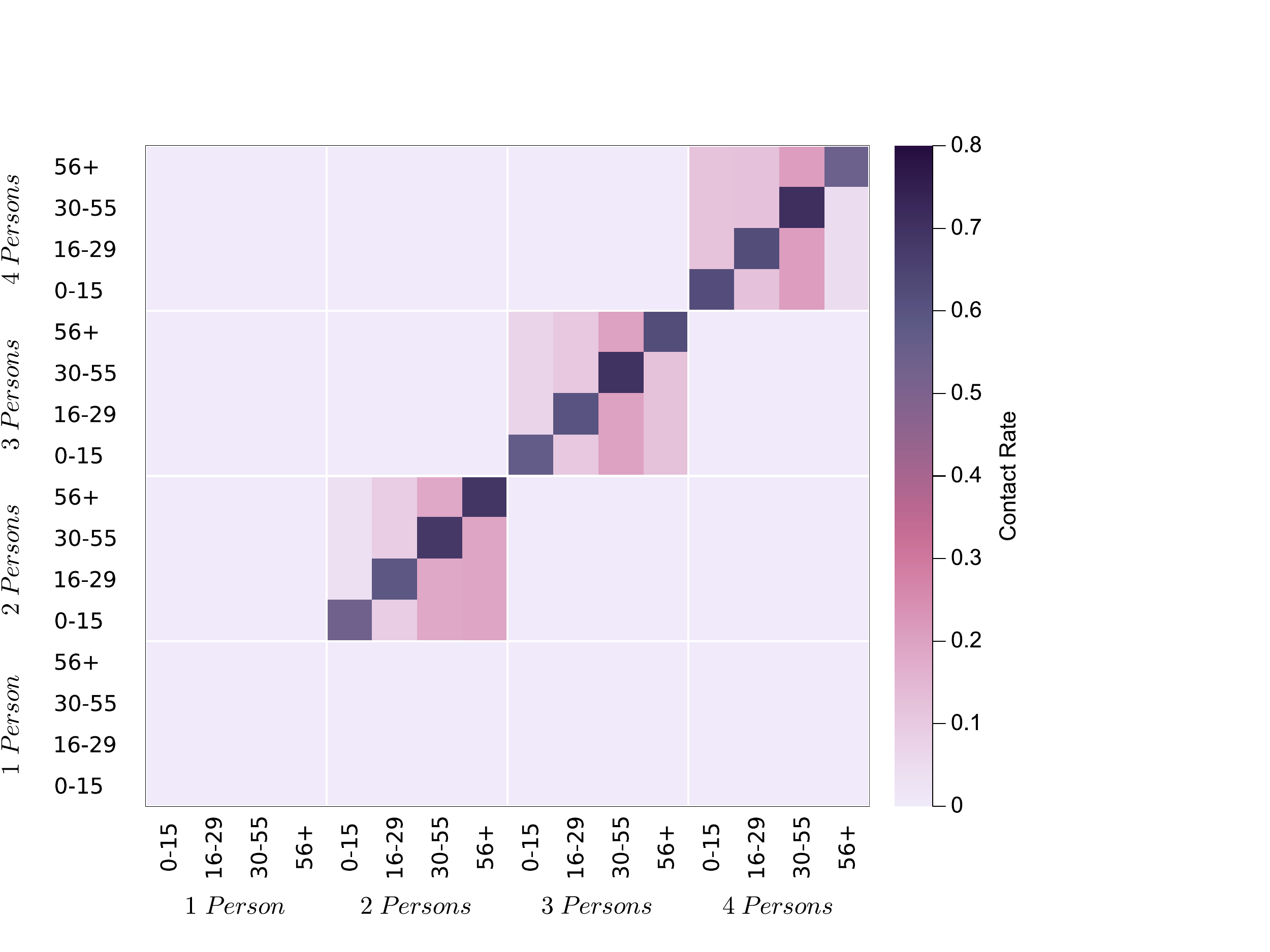}
    \vspace{-0.2cm} 
    \par\center\normalsize (b) within households 
    \label{fig:ch}
  \end{minipage}
  \vspace{-0.0cm}
  \caption{Examples of social contact patterns with four age categories ($0-15$ years, $16-29$ years, $30-55$ years, $56+$ years) and for four different household sizes ($1$ person, $2$ persons, $3$ persons, $4$ persons). (a) A matrix $\bm{C^p}$ for contacts in the public, 
  illustrating interactions across various subpopulations in public environments. (b) A 
  matrix $\bm{C^h}$ for contacts within households, highlighting the interaction patterns among 
  household members. The matrices appear similar across household size categories due 
  to the exemplary choice of uniform household size distributions and within-household 
  contact proportions for all household types.}
  \label{fig:Contact_matrices}
  \vspace{-0.3cm}
\end{figure}

\subsubsection{Contacts within households}

For constructing the contact matrix~$\textbf{C}^h=(c^h_{ij})_{i,j=1,\dots,K}$ for interactions within households, we make two key simplifying assumptions. First, we exclude home visits between households, meaning that contacts at home only occur among individuals living in the same household. Second, interaction intensity within households only depends on age categories rather than household sizes. As a consequence, the contact rate~$c_{ij}^h$ automatically equals zero if subpopulations~$i$ and~$j$ belong to different household size categories, which we express through an indicator~$\Delta_{ij}\in\{0,1\}$ in the following (where the value one refers to equal household size categories).
Furthermore, as a direct consequence of excluding home visits, the contact rate for individuals living in single-person households is set to zero. Figure~\ref{fig:Contact_matrices}(b) displays an example of a within-household contact matrix. 
We apply the approach from Equation~\eqref{matrixconstr}, but this time confined to the scope of each household size category:
\begin{equation}\label{matrix_constr_ch}    c^h_{ij}=\Delta_{ij}\left[\delta_{ij}\epsilon^h_i+(1-\epsilon^h_i)f^h_j\right]
\end{equation}
with
\begin{equation*}
    f^h_j = \frac{(1-\epsilon^h_j)a^h_j N_j}{\sum_{k=1}^K (1-\epsilon^h_k)a^h_k N_k}.
\end{equation*}
In these equations, the variables $\epsilon_j^h$, $f_j^h$ and~$ a_j^h$ retain their previous meanings, but now in the household setting.  The parameter $a_j^h$ represents the number of household contacts available per individual in subpopulation $j$. For a household of size~$l$, this average equals~$l-1$, reflecting the number of other household members available for contact. While household members interact repeatedly throughout the day, we model this as one effective contact per household member per day, with the understanding that the cumulative exposure from cohabitation provides sufficient opportunity for transmission. The public policy parameter is omitted, as lockdown measures do not restrict within-household interactions.

\subsection{SDEs for health states of individuals}
We now describe the probabilities of transition events in our model and from there proceed to SDE approximations. Let $S_k(t)$ and $I_k(t)$ denote the numbers of susceptible and infected individuals in subpopulation $k$ at time $t$, respectively. We first consider them as state components of a Markov jump process, with transitions occurring when susceptible individuals have infection-relevant contacts with infected ones, at rates that depend on the environment (public or household); and when infected individuals recover. Figure~\ref{fig:mjp_structure} illustrates the structure of this Markov jump process.
In the following, we derive expressions of the transition probabilities for infection and recovery. Recall that we assume closed subpopulations, i.\,e.\  individuals remain in their initially assigned subpopulations during the modelled time windows. For notational simplicity, we will omit the time index~$t$ hereafter.

The probability of infection of a susceptible from subpopulation
$k$ during a short time period  $\Delta t$ combines contributions from both public and household settings and is given by
 is given by
\begin{equation*}
    P_k(\bm Y)\Delta t + \circ (\Delta t)
\end{equation*}
with 
\begin{equation}
\label{prob_infection_of_ind}
    P_k(\bm Y)=\alpha\, S_k \sum_{j=1}^K \left[c_{kj}^p +
    \left( 
       \sum_{l=1}^L z_{kl}H_{I,l}
    \right)
    \left( 
       \sum_{l=1}^L z_{kl}\Big(1-\frac{H_{I,l}}{H_{l}}\Big)^{-1}
    \right)
    \frac{c_{kj}^h}{N_{H,k}}\right]\frac{I_j}{N_j},
\end{equation}
where the term \( o(\Delta t) \) represents the probability of multiple events occurring within the infinitesimal time interval \(\Delta t\), which becomes negligible as \(\Delta t\) approaches zero.

To explain Expression~\eqref{prob_infection_of_ind}, we break it down into pieces. We begin with the force of infection on subpopulation~$k$ in the public:
\begin{equation*}
\alpha\, S_k  \sum_{j=1}^K c_{kj}^p \frac{I_j}{N_j}\,.
\end{equation*}
This term accounts for infections resulting from public interactions between susceptibles in subpopulation $k$ and all infected individuals in other subpopulations including~$k$ itself, with a contact rate of $c^p_{kj}$. The parameter~$\alpha$ represents the infection transmission rate that quantifies the efficiency of disease transmission given a contact of sufficient duration to provide transmission opportunity, incorporating biological factors of the pathogen.

\noindent Similarly, the term
\begin{equation*}
    \alpha\,S_k  \sum_{j=1}^K  \frac{c_{kj}^h}{N_{H,k}} \frac{I_j}{N_j} 
\end{equation*}
represents infections resulting from interactions within households between susceptibles of subpopulation~$k$ and infected individuals of subpopulation~$j$. This term captures household transmission with the following components:
\begin{itemize}
    \item $c^h_{kj}$: the household contact rate between subpopulations $k$ and $j$
    \item $\frac{1}{N_{H,k}}$: normalization factor accounting for the fact that $S_k$ susceptible individuals are distributed across $N_{H,k}$ households, converting the contact rate to a per-household basis
    \item $\frac{I_j}{N_j}$: the proportion of infected individuals in subpopulation $j$, representing the baseline infection prevalence
\end{itemize}
The division by $N_{H,k}$ is essential because it scales the household contact rate appropriately: without it, the infection probability would not properly account for how susceptible individuals are distributed across multiple households.

\noindent Eventually, the term 
\begin{equation}
\left( 
       \sum_{l=1}^L z_{kl}H_{I,l}
    \right)
    \left( 
       \sum_{l=1}^L z_{kl}\Big(1-\frac{H_{I,l}}{H_{l}}\Big)^{-1}
    \right)
\label{contact_intensity_multiplier}
\end{equation}
(with components~$z_{kl}$ of the conformability matrix~$\bm{Z}$ introduced before)  serves as a contact intensity multiplier that concentrates household transmission risk. It increases the infection probability for individuals in infected households relative to the population average, reflecting the reality that once infection enters a household, household members face elevated risk. In simple terms, if the fraction of uninfected households $1-H_{I,l}/H_{l}$ is high, its reciprocal will be a small number, effectively decreasing the impact of within-household contacts. On the other hand, if most households have at least one infected person (meaning the fraction of uninfected households is low), the reciprocal will be large, increasing the impact of household contacts on the probability of infection \citep{Bayham2016}.

\noindent To simplify the notations, we define
\begin{equation}
    c_{kj}=c_{kj}^p  +
   \left(\sum_{l=1}^L z_{kl} H_{I,l}\right) \left(\sum_{l=1}^L z_{kl}\Big(1-\frac{H_{I,l}}{H_{l}}\Big)^{-1}\right)
    \frac{c_{kj}^h}{N_{H,k}}\,,
    \label{formula_ckj}
\end{equation}
leading to 
\begin{equation*}
\label{prob_infection_of_ind2}
    P_k(\bm Y)=\alpha\, S_k \sum_{j=1}^K c_{kj}\frac{I_j}{N_j}.
\end{equation*}

The probability of recovery of an infected individual from subpopulation~$k$ during the time period $\Delta t$ is modelled as
\begin{equation*}
    \Pi_k(\bm Y)\Delta t +\circ(\Delta t)
\end{equation*}
with
\begin{equation}
    \Pi_k(\bm Y) = \beta_k I_k
\end{equation}
where $\beta_k$ is the recovery rate for subpopulation $k$, representing the inverse of the respective average infectious period.

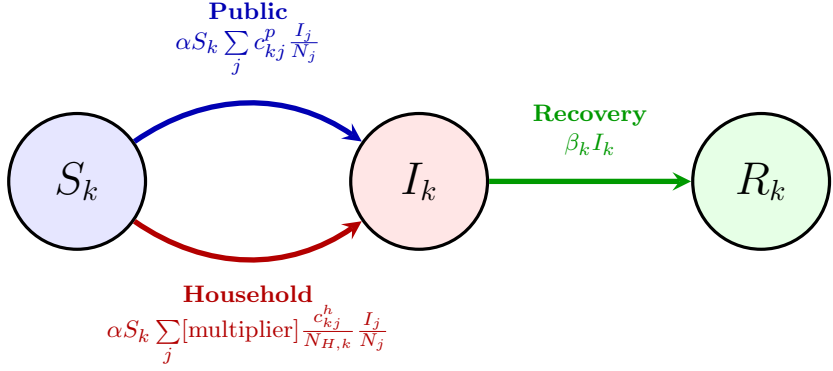
\begin{figure}[htbp]
\centering
\begin{tikzpicture}[
    node distance=4.5cm,
    state/.style={circle, draw=black, very thick, minimum size=1.8cm, font=\Large\bfseries},
    arrow/.style={->, >=stealth, very thick}
]

\node[state, fill=blue!10] (S) {$S_k$};
\node[state, fill=red!10, right of=S] (I) {$I_k$};
\node[state, fill=green!10, right of=I] (R) {$R_k$};

\draw[arrow, blue!70!black, line width=2pt, bend left=35] (S) to node[above=0.2cm, align=center, font=\small] {
    \textcolor{blue!70!black}{\textbf{Public}}\\
    $\alpha S_k \sum\limits_{j} c_{kj}^p \frac{I_j}{N_j}$
} (I);

\draw[arrow, red!70!black, line width=2pt, bend right=35] (S) to node[below=0.2cm, align=center, font=\small] {
    \textcolor{red!70!black}{\textbf{Household}}\\
    $\alpha S_k \sum\limits_{j} [\text{multiplier}] \frac{c_{kj}^h}{N_{H,k}} \frac{I_j}{N_j}$
} (I);

\draw[arrow, green!60!black, line width=2pt] (I) to node[above=0.2cm, align=center, font=\small] {
    \textcolor{green!60!black}{\textbf{Recovery}}\\
    $\beta_k I_k$
} (R);

\end{tikzpicture}
\caption{Markov jump process for subpopulation $k$ showing transitions through susceptible ($S_k$), infected ($I_k$), and recovered ($R_k$) states. Infection occurs via public or household transmission, with the household pathway incorporating a contact intensity multiplier (Equation~\ref{contact_intensity_multiplier}).}
\label{fig:mjp_structure}
\end{figure}

To derive SDEs from the Markov jump process, we follow the diffusion approximation procedure outlined in \cite{Fuchs2013}. First, we normalize the state variables by the system sizes to ensure scale invariance and better capture stochastic fluctuations:

\begin{align*}
      \bm x&=\left(s_1, i_1,\dots,s_K, i_K\right)^T=\left(\frac{S_1}{N_1},\frac{I_1}{N_1},\dots,\frac{S_K}{N_K},\frac{I_K}{N_K}\right)^T\\
\bm h&=\left(h_{S,1},h_{I,1},\dots,h_{S,L},h_{I,L}\right)^T=\left(\frac{H_{S,1}}{H_1},\frac{H_{I,1}}{H_1},\dots,\frac{H_{S,L}}{H_{L}},\frac{H_{I,L}}{H_{L}}\right)^T,
\end{align*} 
leading to $\bm y = (\bm x^T, \bm h^T)^T$.
Following this normalization, for $k=1,\dots,K$, the main part of the infection probability (leaving out~$\Delta t$ and $\circ(\Delta t)$ for simplicity) transforms to its intensive form
$$
p_k(\bm y)=\frac{1}{N_k} P_k(\bm Y) =\alpha\, s_k \sum_{j=1}^K c_{kj}i_j,
$$
and the main part of the recovery probability becomes
\begin{align*}
\pi_k(\bm y)=\frac{1}{N_k}\Pi_k(\bm Y)=\beta_k i_k.
\end{align*}

In what follows, we denote by $\bm{y}(t)$ the time-dependent variable where the temporal argument is explicitly shown, while we simplify the notation to $\bm{y}$ when the time dependence is contextually clear.
Through the diffusion approximation of the Markov jump process, we obtain the SDE describing the proportion of susceptibles in population $k$ over time, $s_k$:
\begin{equation}
\begin{cases}
ds_k(t)=\mu_k^S(\bm{y}(t))dt+\sigma_k^{SS}(\bm{y}(t))d B_{k,1}(t)\\
s_k(t_0)=s_k^0
\end{cases}
\end{equation}
with 
\begin{eqnarray*}
    \mu_k^S(\bm y)=- p_k(\bm y)
=-\alpha\,s_k \sum_{j=1}^K  c_{kj} i_j\,
\end{eqnarray*}
and
\begin{equation*}
    \sigma^{SS}_{k}(\bm y)=\sqrt{\frac{p_k(\bm y)}{N_k}}\,.
\end{equation*}

The SDE for the proportion of infected individuals in subpopulation~$k$ results as:
\begin{equation}
\begin{cases}
di_k(t)=\mu_k^I(\bm{y}(t))dt+\left[\sigma_k^{SI}(\bm{y}(t))dB_{k,1}(t)+\sigma_k^{II}(\bm{y}(t))d B_{k,2}(t)\right]\\
i_k(t_0)=i_k^0.
\end{cases}
\end{equation}
The drift term $\mu_k^I(\bm y)$ represents the rate of change in the infected proportion, given by the difference between the infection rate $p_k(\bm y)$ and the recovery rate $\pi_k(\bm y)$:
$$
\mu_k^I(\bm y)=p_k(\bm y)-\pi_k(\bm y)
=\alpha\,s_k\sum_{j=1}^K c_{kj} i_j-\beta_k i_k
$$
and the diffusion coefficients \citep{Fuchs2013}
\begin{align*}
    \sigma^{II}_{k}(\bm y)&=\sqrt{\frac{\pi _k(\bm y)}{N_k}}, &
    \sigma^{SI}_{k}(\bm y)&=-\sqrt{\frac{p_k(\bm y)}{N_k}}\,.
\end{align*}
Combining the SDEs for both susceptible and infected proportions, and incorporating their respective drift and diffusion terms, we obtain the complete system of SDEs describing the health states of individuals:
\begin{equation}\label{sdep}
\begin{cases}
ds_k(t)=\mu_k^S(\bm{y}(t))dt+\sigma_k^{SS}(\bm{y}(t))d B_{k,1}(t)\\
di_k(t)=\mu_k^I(\bm{y}(t))dt+\left[\sigma_k^{SI}(\bm{y}(t))d B_{k,1}(t)+\sigma_k^{II}(\bm{y}(t))d B_{k,2}(t)\right]\\
(s_k,i_k)(t_0)=(s_k^0,i_k^0),
\end{cases}
\end{equation}
where $k=1,\dots,K$, and the recovered proportion $r_k$ can be deduced from the conservation equation $r_k = 1-s_k-i_k$. 
The equations contain an initial time~$t_0$, initial states~$s_k^0$ and~$i_k^0$, and independent standard Brownian motions~$B_{k,1}(t)$ and~$B_{k,2}(t)$ for $t \geq t_0$, where $B_{k,1}(t_0) = B_{k,2}(t_0) = 0$ for all $k=1,\ldots,K$.

\subsection{SDEs for health states of households}

Households contain individuals from different subpopulations. For example, in a family, parents and children belong to different age categories but the same household size category~$l$. The conformability matrix $\bm Z$ enables extraction of all individuals of a specific age group from a particular household size category: multiplying $\bm S$ or $\bm I$ by the $l$-th column of $\bm Z$ yields all susceptibles or infectives in households of size~$l$. For three age categories and two household sizes, the conformity matrix can have the form:
\[\label{conform}
\bm Z=
\begin{array}{c|cc}
 & \text{size cat.\ 1} & \text{size cat.\ 2} \\ \hline
\text{(age cat.\ 1, size cat.\ 1)} & 1 & 0  \\
\text{(age cat.\ 2, size cat.\ 1)} & 1 & 0  \\
\text{(age cat.\ 3, size cat.\ 1)} & 1 & 0  \\
\text{(age cat.\ 1, size cat.\ 2)} & 0 & 1  \\
\text{(age cat.\ 2, size cat.\ 2)} & 0 & 1  \\
\text{(age cat.\ 3, size cat.\ 2)} & 0 & 1  \\
\end{array}
\]
Mathematically, the elements of $\bm Z$ are defined as
\begin{equation*}
    z_{kl}=
    \begin{cases}
        1, & \text{if } (l-1)A < k \leq lA\\
        0, & \text{otherwise}
    \end{cases}
\end{equation*}
for $k=1,\dots,K$ and $l=1,\dots,L$.

To derive the SDE model for household states, we follow the same diffusion approximation procedure as for individuals. For brevity, we work directly with normalized variables: $h_{S,l} = H_{S,l}/H_l$ and $h_{I,l} = H_{I,l}/H_l$ for household proportions, along with the individual proportions $s_k$ and $i_k$ defined earlier. The transition rate for a susceptible household becoming infected depends on: (1) the probability of any of its members encountering an infected individual in public spaces, and (2) the subsequent transmission of infection to that household member.

The probability of a susceptible household from size category~$l$ becoming infected during the time period $\Delta t$ is 

\begin{equation*}
p'_l(\bm y)\Delta t+\circ (\Delta t)
\end{equation*}
for $l=1,\dots,L$ with
\begin{equation}
p'_l(\bm y)=\alpha \frac{h_{S,l}}{H_{l}}\sum_{k=1}^K z_{kl}N_{k}s_{k}\sum_{j=1}^Kc_{kj}^p i_j.
\end{equation}
This formula conveys the idea that the transmission of infections to susceptible households exclusively occurs through their engagements in public spaces by definition. Specifically, the term $z_{kl}N_{k}s_{k}$ serves to extract those susceptible individuals from a household of size category~$l$ who encounter infected persons from various subpopulations in public spaces, characterized by a contact rate of $c_{kj}^p$. The multiplication by $h_{S,l}/H_{l}$ further refines this interaction by proportionally adjusting for the share of susceptibles within a given household size category~$l$. This construction aligns with our earlier definition of a household infection event as the initial infection occurring within an entirely susceptible household. 

Through the diffusion approximation, we obtain the SDE for the proportion of susceptible households of size $l=1,\dots,L$:
\begin{equation}
\begin{cases}
dh_{S,l}(t)=\mu^{S,h}_{l}(\bm{y}(t))dt+\sigma^{SS,h}_{l}(\bm{y}(t))d B'_{l,1}(t)\\
h_{S,l}(t_0)=h_{S,l}^0,
\end{cases}
\end{equation}
where 
\begin{equation*}
\mu_{l}^{S,h}(\bm y)=- p'_l(\bm y), \text{  }\quad \sigma^{SS,h}_{l}(\bm{y})=\sqrt{\frac{p'_l(\bm y)}{H_{l}}}\,.
\end{equation*}
For the recovery process, the probability of an infected household of size $l$ recovering during time period $\Delta t$ is
\begin{equation*}
\pi'_l(\bm y)\Delta t +\circ(\Delta t)
\end{equation*}
with
\begin{equation}
\pi'_l(\bm y) = \nu_l h_{I,l}
\end{equation}
where $\nu_l$ is the household recovery rate for size category~$l$.
The corresponding SDE for the proportion of infected households follows from the diffusion approximation:
\begin{equation}
\begin{cases}
dh_{I,l}(t)=\mu^{I,h}_{l}(\bm{y}(t))dt+[\sigma^{SI,h}_{l}(\bm{y}(t))d B'_{l,1}(t)+\sigma^{II,h}_{l}(\bm{y}(t))d B'_{l,2}(t)]\\
h_{I,l}(t_0)=h_{I,l}^0.
\end{cases}
\end{equation}
The drift term represents the rate of change:
\begin{equation*}
\mu^{I,h}_{l}(\bm y)=p'_l(\bm y)-\pi'_l(\bm y)
=\alpha \frac{h_{S,l}}{H_{l}}\sum_{k=1}^K z_{kl}N_{k}s_{k}\sum_{j=1}^Kc_{kj}^p i_j-\nu_l h_{I,l}
\end{equation*}
with diffusion terms:
\begin{align*}
\sigma^{SI,h}_{l}(\bm y)=-\sqrt{\frac{p'_ l(\bm y)}{H_{l}}}, 
\qquad\sigma^{II,h}_{l}(\bm y)&=\sqrt{\frac{\pi'_l(\bm y)}{H_{l}}}.
\end{align*}
Together, the SDE system describing the health states for households of size categories $l=1,\dots,L$ is defined as
\begin{equation}\label{sdeh_h}
    \begin{cases}
        dh_{S,l}(t)=\mu^{S,h}_{l}(\bm{y}(t))dt+\sigma^{SS,h}_{l}(\bm{y}(t))d B'_{l,1}(t)\\ 
        dh_{I,l}(t)=\mu^{I,h}_{l}(\bm{y}(t))dt+\left[\sigma^{1,h}_{l}(\bm{y}(t))d B'_{l,2}(t)+\sigma^{II,h}_{l}(\bm{y}(t))d B'_{l,II}(t)\right]\\
        (h_{S,l},h_{I,l})(t_0)=(h_{S,l}^0,h_{I,l}^0).
    \end{cases}    
\end{equation}
The equations are defined for $t \geq t_0$ with initial states~$h_{S,l}^0$ and~$h_{I,l}^0$, driven by independent standard Brownian motions~$B'_{l,1}(t)$ and~$B'_{l,2}(t)$ where $B'_{l,1}(t_0) = B'_{l,2}(t_0) = 0$ for all $l=1,\ldots,L$. The collection $\{B'_{l,1}(t), B'_{l,2}(t)\}_{l=1}^L$ consists of $2L$ Brownian motions that are mutually independent and independent of the individual-level Brownian motions $\{B_{k,1}(t), B_{k,2}(t)\}_{k=1}^K$.

The final system of SDEs coupling the state variables for individuals and households is a combination of  (\ref{sdep}) and (\ref{sdeh_h}) is
\begin{equation}\label{sdeh}
    \begin{cases}
    ds_k(t)=\mu_k^S(\bm{y}(t))dt+\sigma_k^{SS}(\bm{y}(t))d B_{k,1}(t)\\ 
    di_k(t)=\mu_k^I(\bm{y}(t))dt+\left[\sigma_k^{SI}(\bm{y}(t))d B_{k,1}(t)+\sigma_k^{II}(\bm{y}(t))d B_{k,2}(t)\right]\\
    dh_{S,l}(t)=\mu^{S,h}_{l}(\bm{y}(t))dt+\sigma^{SS,h}_{l}(\bm{y}(t))d B'_{l,1}(t)\\ 
    dh_{I,l}(t)=\mu^{I,h}_{l}(\bm{y}(t))dt+\left[\sigma^{SI,h}_{l}(\bm{y}(t))dB'_{l,1}(t)+\sigma^{II,h}_{l}(\bm{y}(t))d B'_{l,2}(t)\right]\\
        (s_k,i_k,h_{S,l},h_{I,l})(t_0)=(s_k^0,i_k^0,h_{S,l}^0,h_{I,l}^0)\\  
    \end{cases}    
\end{equation}
 with 
\begin{equation}   
\begin{array}{@{}rcl@{\qquad}rcl@{}}
\mu_k^S(\bm{y}) &=& -\alpha\, s_k \displaystyle\sum_{j=1}^K c_{kj}\, i_j
  & \mu^{S,h}_{l}(\bm{y}) &=& -\alpha \dfrac{h_{S,l}}{H_{l}} \displaystyle\sum_{k=1}^K z_{kl}\,N_{k}\, s_{k}\,\displaystyle\sum_{j=1}^K c_{kj}^p\, i_j \\[3mm]
\mu_k^I(\bm{y}) &=& \alpha\, s_k \displaystyle\sum_{j=1}^K c_{kj}\, i_j - \beta_k\, i_k,
  & \mu^{I,h}_{l}(\bm{y}) &=& \alpha \dfrac{h_{S,l}}{H_{l}} \displaystyle\sum_{k=1}^K z_{kl}\,N_{k}\, s_{k}\,\displaystyle\sum_{j=1}^K c_{kj}^p\, i_j - \nu_l\, h_{I,l} \\[3mm]
\sigma_{k}^{SS}(\bm{y}) &=& \sqrt{\dfrac{\alpha }{N_k}\, s_k\displaystyle\sum_{j=1}^K c_{kj}\, i_j},
  & \sigma^{SS,h}_{l}(\bm{y}) &=& \sqrt{\dfrac{\alpha}{H_{l}}\,\dfrac{h_{S,l}}{H_{l}}\displaystyle\sum_{k=1}^K z_{kl}\,N_{k}\, s_{k}\,\displaystyle\sum_{j=1}^K c_{kj}^p\, i_j} \\[3mm]
\sigma_{k}^{II}(\bm{y}) &=& \sqrt{\dfrac{\beta_k\, i_k}{N_k}},
  & \sigma^{II,h}_{l}(\bm{y}) &=& \sqrt{\dfrac{\nu_l\, h_{I,l}}{H_{l}}} \\[3mm]
\sigma_{k}^{SI}(\bm{y}) &=& -\sigma_{k}^{SS}(\bm{y}),
  & \sigma^{SI,h}_{l}(\bm{y}) &=& -\sigma^{SS,h}_{l}(\bm{y})
\end{array}
\label{SDE_terms}
\end{equation}
for $k=1,\dots,K$ and $l=1,\dots,L$.

The existence and uniqueness of a strong solution to the system of SDEs~\eqref{sdeh} requires careful treatment due to coefficient degeneracies: the diffusion terms contain non-Lipschitz square roots at zero, and the drift has a potential singularity at $h_{I,l} = 1$ from the term $(1-h_{I,l})^{-1}$ in $c_{kj}$. We establish existence and uniqueness using the Yamada-Watanabe pathwise uniqueness theorem \citep{Mao2008} for non-Lipschitz diffusions, combined with a localization argument. We prove that solutions remain in the physical domain $\mathcal{D} = \{(s_k,i_k,h_{S,l},h_{I,l}) : s_k,i_k,h_{S,l},h_{I,l} \geq 0, s_k+i_k \leq 1, h_{S,l}+h_{I,l} \leq 1\}$ almost surely, ensuring global existence and preservation of physical constraints. The detailed proof is in Appendix~\ref{Appendix C: Existence}.

\section{Multi-Type Branching Process Approximation}
\label{sec:branching_process_approx}

To derive the basic reproduction number $\mathcal{R}_0$ for our stochastic epidemic model, we employ the branching process approximation framework established by \cite{allen2017primer} (who use the term continuous-time Markov chains instead of MJP). Since our model is formulated as a system of SDEs derived from an underlying MJP through diffusion approximation, we apply this framework to the embedded discrete jump structure. In the stochastic setting, $\mathcal{R}_0$ determines the probability of disease extinction: when $\mathcal{R}_0 \leq 1$, the infection dies out soon with probability one, while when $\mathcal{R}_0 > 1$, there exists a positive probability of a major outbreak, though extinction remains possible due to stochastic fluctuations.

Following \cite{allen2017primer}, we exploit the result that near the disease-free equilibrium (DFE) when the number of infected individuals is small, the MJP epidemic model can be approximated by a multi-type branching process. We define the DFE in an epidemic model as a balanced state where the disease is absent from the population \citep{Brauer2019}. Mathematically, it corresponds to an equilibrium where all drift terms equal zero and all infection compartments are empty. At the DFE, the population consists solely of susceptible individuals and, depending on the model structure, recovered individuals.
To determine the DFE of the SDE system (\ref{sdeh}), we set all the drift terms to zero and require all infection compartments to be empty. The unique DFE is given by
\begin{align*}
    \bm{y}^* &= (s_1^*,i_1^*, \dots, s_K^*, i_K^*, h_{S,1}^*,h_{I,1}^*, \dots, h_{S,L}^*,h_{I,L}^*)\\
    & = (1, 0,\ldots, 1, 0,  1,0,\ldots, 1,0).
\end{align*}

In \citet{allen2017primer}, the approximation by multi-type branching processes is justified because, \textit{near} the DFE, the susceptible populations remain approximately at their equilibrium values ($s_k \approx 1$), effectively linearizing the nonlinear infection dynamics. As Allen demonstrates, this linearization transforms the transition rates into forms that depend linearly on the number of infectives, precisely the mathematical structure of a branching process.

Our methodology follows a systematic two-step approach. First, in Section~\ref{sec:linearization}, we linearize the transition rates of our MJP near the DFE, showing that household-mediated transmission effects become negligible in this limit. This linearization yields transition rates that depend only on public contact patterns, with household structure effects vanishing near the DFE. Second, in Section~\ref{sec:branching}, we construct the multi-type branching process using these linearized rates and apply standard branching process theory \citep{Mode1971, athreya2017branching} to derive $\mathcal{R}_0$.

This approach offers several advantages. It provides a rigorous theoretical foundation for computing $\mathcal{R}_0$ that applies to both the underlying MJP and its SDE approximation. Moreover, it avoids the mathematical complexity of analyzing SDEs directly while leveraging the well-established theory of branching processes. The resulting expression for $\mathcal{R}_0$ governs the threshold behavior of the epidemic outlined above.

\subsection{Linearization Near the Disease-Free Equilibrium}
\label{sec:linearization}

At the considered DFE~$\bm{y}^*$, the state variables satisfy $s_k = 1$ and $i_k = 0$ for all $k \in \{1,\ldots,K\}$, and $h_{S,l} = 1$ and $h_{I,l} = 0$ for all $l \in \{1,\ldots,L\}$.

\subsubsection{Linearization of Contact Rates}

Near the DFE, the coupling term in Equation~\eqref{formula_ckj} simplifies considerably. Since $H_{I,l} \approx 0$, we have $(1-H_{I,l}/H_l)^{-1} \approx 1$ and $\sum_{l=1}^L z_{kl}H_{I,l} \approx 0$. Consequently, the contact rate reduces to its public component:
\begin{equation}
c_{kj} \approx c_{kj}^p \quad \text{for } H_{I,l} \approx 0.
\label{eq:ckj_linearized_dfe}
\end{equation}

This linearization reveals that household-mediated transmission effects are negligible compared to public transmission near the DFE, justifying their omission in the branching process approximation for threshold analysis \citep{Ball1995,Pellis2012}.

\subsubsection{Linearization of Transition Rates}
\label{sec:linear_trans_rates}

In the branching process, we model the evolution and count the number of infected individuals. The transition rates from the MJP near the DFE~$\bm{y}^*$ are as follows. For an infected individual from subpopulation~$k\in\{1,\ldots,K\}$, we have two primary events:
\textbf{Recovery event:} The considered individual recovers at rate
\begin{equation*}
\lambda_k^\text{(rec)} = \beta_k.
\end{equation*}
\textbf{Infection event:} The considered individual passes the infection to a  susceptible individual from subpopulation~$j\in\{1,\ldots,K\}$  through contagious contact. Near the DFE, where $s_j \approx 1$, the rate of this event is
\begin{equation*}
\lambda_{k \to j}^{\text{(inf)}} = \alpha c_{kj}^p s_j \approx \alpha c_{kj}^p \quad \text{for each } j \in \{1,\ldots,K\}.
\end{equation*}

The \textbf{total event} rate for an infected individual from subpopulation~$k$ is therefore
 \begin{equation*}
\lambda_k = \lambda_k^\text{(rec)} + \sum_{j=1}^K \lambda_{k \to j}^\text{(inf)} = \beta_k + \alpha \sum_{j=1}^K c_{kj}^p.
\end{equation*}

\subsection{Construction of the Multi-Type Branching Process}
\label{sec:branching}

We construct a continuous-time multi-type branching process \citep{Mode1971} where types correspond to the sets of infected individuals in the $K$ subpopulations. Each infected individual of type~$k \in \{1,\ldots,K\}$ can undergo one of the two kinds of events described in Section~\ref{sec:linear_trans_rates}. 

In the branching process approximation, each event contributes a term to the probability generating function (PGF). Following the continuous-time branching process framework \citep{Mode1971}, the PGF for type $k$ is constructed as:
\begin{equation}
f_k(\mathbf{u}) = \sum_{\text{events}} \frac{\text{rate of event}}{\text{total event rate}} \times \text{(offspring pattern)},
\label{eq:pgf_construction}
\end{equation}
where $\mathbf{u} = (u_1,\ldots,u_K)^T \in [0,1]^K$, the sum runs over all events for the considered individual, and the offspring pattern encodes the state after the event occurs via the exponent of the components of~$\mathbf{u}$.
For our epidemic model, where the infecting individual remains in the population \citep{Kendall1948,Ball1983}, 
the kinds of possible events have been described above: there is either recovery, resulting in zero offspring since the considered individual exits the infectious state, or infection, resulting in one offspring of type~$k$ (the parent remains in the infectious state) and one offspring of type~$j$ (newly created).

Combining these contributions as in Equation~\eqref{eq:pgf_construction} yields the complete (approximate) PGF:
\begin{equation}
f_k(\mathbf{u}) 
= \frac{\lambda_k^\text{(rec)}}{\lambda_k} u_k^0 + \sum_{j=1}^K \frac{\lambda_{k\rightarrow j}^\text{(inf)}}{\lambda_k} u_k^1 u_j^1
= \frac{\beta_k}{\lambda_k} + \sum_{j=1}^K \frac{\alpha c_{kj}^p}{\lambda_k} u_k u_j
\label{eq:pgf_branching}
\end{equation}
with 
$\lambda_k = \beta_k + \alpha \sum_{j=1}^K c_{kj}^p$ as derived above.

\subsection{Mean Offspring Matrix and Basic Reproduction Number}

The mean offspring matrix $\mathbf{M} = (m_{kj})_{k,j=1,\ldots,K}$ has the following entries:
\begin{eqnarray*}
m_{kj} = \frac{\partial f_k(\mathbf{u})}{\partial u_j}\bigg|_{\mathbf{u}=\mathbf{1}} &=& \begin{cases}
\frac{\alpha c_{kj}^p}{\lambda_k} & \text{if } k \neq j \\[2mm]
\frac{2\alpha c_{kj}^p}{\lambda_k} & \text{if } k = j
\end{cases}
\label{eq:mean_matrix_entries}
\end{eqnarray*}
with~$\mathbf{1}=(1,\ldots,1)$ and~$\mathbf{C}^p=(c_{kj}^p)_{k,j=1,\ldots,K}$ the public contact matrix. 
The diagonal entries of~$\mathbf{M}$ account for the parent remaining in the population, a characteristic feature of epidemic branching processes \citep{Ball1983}.  In matrix form, this can be written as:
\begin{equation}
\mathbf{M} = \alpha \cdot \text{diag}(\lambda_1^{-1}, \ldots, \lambda_K^{-1}) \cdot (\mathbf{C}^p+\text{diag}(\mathbf{C}^p)).
\label{eq:mean_matrix_form}
\end{equation}
The basic reproduction number, following the next-generation matrix approach \citep{Diekmann1990}, is given by the spectral radius of the mean offspring matrix:
\begin{equation*}
\mathcal{R}_0 = \rho(\mathbf{M}).
\end{equation*}
A closed-form expression for $\mathcal{R}_0$ in terms of the model parameters is generally not available for $K \geq 3$ subpopulations. The spectral radius is defined as the largest eigenvalue in modulus, $\gamma$, which requires solving the characteristic polynomial \linebreak $\det(\mathbf{M} - \gamma \mathbf{I}) = 0$. For a $K \times K$ matrix, this yields a polynomial of degree $K$, and by the Abel-Ruffini theorem \citep{tignol2015galois}, no general algebraic solution exists for polynomials of degree five or higher. Even for $K = 3$ or $K = 4$, while solutions exist in principle, they involve complicated expressions in the matrix entries that provide little analytical insight given the already complex dependence of $c_{kj}^p$ on the underlying epidemiological and behavioral parameters.

In practice, $\mathcal{R}_0$ is therefore computed numerically. Given a parameter configuration, the computation proceeds as follows: (i) construct the public contact matrix $\mathbf{C}^p$ from the activity levels, mixing preferences, and policy restrictions; (ii) compute the total event rates $\lambda_k = \beta_k + \alpha \sum_{j=1}^K c_{kj}^p$ for each subpopulation; (iii) assemble the mean offspring matrix $\mathbf{M}$ according to Equation~\eqref{eq:mean_matrix_form}; and (iv) compute the eigenvalues of~$\mathbf{M}$ and return the largest real part. Standard numerical linear algebra routines perform this eigenvalue computation with a complexity of $O(K^3)$ \citep[Ch.~7]{golub2013matrix}, making the procedure efficient even for models with many subpopulations. The sensitivity analysis in Section~\ref{sec:GSA} employs this numerical approach to systematically explore how $\mathcal{R}_0$ responds to variations in the model parameters.

While household-mediated transmission is negligible for determining $\mathcal{R}_0$ near the DFE, it becomes significant as the epidemic progresses. The enhanced within-household transmission, captured by the term $(1-H_{I,l}/H_l)^{-1}$ in Equation~\eqref{formula_ckj}, creates positive feedback that accelerates epidemic spread away from the DFE. However, for the determination of $\mathcal{R}_0$ and initial epidemic growth rates, the first-order approximation presented here is both sufficient and exact in the limit as the system approaches the DFE.

\section{Simulation Studies} \label{sec:simulations}
This section explores how our household-structured epidemic model captures real-world disease transmission dynamics and assesses its utility for public health planning. Through detailed numerical simulations using both deterministic and stochastic versions of our model, our analysis is structured into four main components, each designed to illustrate different aspects of the model's behavior and epidemiological implications. First, we conduct a comparative analysis using the deterministic versions of our model against simpler variants (non-structured, age-only, and household-only models; Section~\ref{sec:comparison_othermodels}). We choose the deterministic framework for this comparison as it allows us to focus on the fundamental structural differences between models by examining their mean-field behavior. The second component examines the relative contributions of household and public contacts in disease transmission, providing insights into the interaction between these distinct transmission pathways (Section~\ref{sec:transmission_pathways}). The third component is an investigation of how public health interventions might modulate household exposure intensity, capturing the natural randomness in disease transmission processes (Section~\ref{sec:effect_public_policy}). Afterwards, we perform a global sensitivity analysis to identify which model parameters most significantly influence the basic reproduction number $\mathcal{R}_0$ (Section~\ref{sec:GSA}). Finally, we investigate the impact of demographic stochasticity by comparing deterministic (ODE) and stochastic (SDE) implementations of our model, revealing how noise amplification varies with household size and age structure, and quantifying the resulting variability in epidemic trajectories and peak timing across different subpopulations (Section~\ref{sec: stoch_effect}).

In our simulations, we make several key assumptions regarding population distribution, household composition, and initial infection rates. While acknowledging that population demographics significantly influence simulation outcomes, we aim to minimize these effects to focus on the core components of the model and their relative impacts on overall system behavior. To achieve this, we implement a balanced population distribution across all subpopulations, with equal numbers in each category, except for the subpopulation representing children in single households, which is set to zero to reflect the reality that minors do not typically live alone. We recognize that alternative scenarios, such as age-skewed populations or varying household size distributions in urban versus rural setting, could be subject of future studies. 

For simplification, we assume that each 'individual' in our model represents an entire household. When we track an individual of age $a$ in household size category $l$, this individual serves as a representative for their entire household of size $l$. This representation does not eliminate household structure from the model; rather, the effects of other household members are captured through the household-level state variables ($H_{S,l}$, $H_{I,l}$) and the within-household transmission terms in the contact rate $c_{kj}$. Under this assumption, no two individuals in the simulation share the same physical household, which implies that $N_{H,k} = N_k$ (the number of households containing individuals from subpopulation $k$ equals the number of individuals in that subpopulation).

Our simulation study uses contact data from the POLYMOD survey \citep{Mossong2020}, a large-scale study of social contact patterns conducted across eight European countries in 2005-2006. This dataset has become a standard reference in epidemiological modeling due to its comprehensive information about household compositions and participant demographics. The survey captured both physical and non-physical contacts, recording the age, location, and duration of each interaction. Additionally, it provides in-depth insights into the daily social interactions of the participants, which we use to derive the average contact rates $a_i^p$ for each subpopulation. These rates are incorporated into the construction of our public contact matrix $\bm{C^p}$ according to Equation~\eqref{matrixconstr}. In our study, the population is categorized into $A=4$ age categories: $(0-15)$, $(16-29)$, $(30-55)$, and $56+$ years, as well as $L=4$ household sizes: 1, 2, 3 and~4. Consequently, this results in a total of $K=L\times A=16$ distinct subpopulations; the first subpopulation, however, is left empty since the data does not contain children in single households. 
Table~\ref{tab:init} summarizes all parameter values and initial conditions used throughout our simulations, unless otherwise specified. These baseline values serve as our reference point for exploring the model's behavior under various scenarios.

\begin{table}[htbp]
\centering
\caption{Initial Parameter Values for the Simulation Study}
\begin{tabular}{|l|p{0.4\textwidth}|p{0.3\textwidth}|}
\hline
\textbf{Parameter} & \textbf{Value(s)} & \textbf{Description} \\
\hline
$A$ & 4 & Number of age categories \\
\hline
$L$ & 4 & Number of household sizes \\
\hline
$K$ & 16 & Total number of subpopulations ($A \times L$) \\
\hline
$\alpha$ & 0.2 & Transmission rate \\
\hline
$\epsilon$ & \begin{tabular}[c]{@{}l@{}}
Age 0-15: (0.5, 0.5, 0.5, 0.5)\\
Age 16-29: (0.5, 0.5, 0.5, 0.5)\\
Age 30-55: (0.5, 0.5, 0.5, 0.5)\\
Age 56+: (0.5, 0.5, 0.5, 0.5)
\end{tabular} & Within-subpopulation contact preference \\
\hline
$\omega$ & \begin{tabular}[c]{@{}l@{}}
Age 0-15: (1.0, 1.0, 1.0, 1.0)\\
Age 16-29: (1.0, 1.0, 1.0, 1.0)\\
Age 30-55: (1.0, 1.0, 1.0, 1.0)\\
Age 56+: (1.0, 1.0, 1.0, 1.0)
\end{tabular} & Public policy parameter \\
\hline
$\beta$  & \begin{tabular}[c]{@{}l@{}}
Age 0-15: 0.200\\
Age 16-29: 0.167\\
Age 30-55: 0.143\\
Age 56+: 0.071
\end{tabular} & Recovery rates for each age category \\
\hline
$\nu$ & \begin{tabular}[c]{@{}l@{}}
Size 1: 0.11\\
Size 2: 0.05\\
Size 3: 0.04\\
Size 4: 0.02
\end{tabular} & Recovery rates for each household size \\
\hline
$N$ & \begin{tabular}[c]{@{}l@{}}
Age 0-15: (0, 333, 333, 333)\\
Age 16-29: (333, 333, 333, 333)\\
Age 30-55: (333, 333, 333, 333)\\
Age 56+: (333, 333, 333, 333)
\end{tabular} & Number of individuals in each subpopulation \\
\hline
$H$ & \begin{tabular}[c]{@{}l@{}}
Size 1: 999\\
Size 2: 1,332\\
Size 3: 1,332\\
Size 4: 1,332
\end{tabular} & Number of households of each size \\
\hline
$a^p$ & \begin{tabular}[c]{@{}l@{}}
Age 0-15: (NA, 2.77, 2.91, 1.3)\\
Age 16-29: (4.7, 3.39, 2.44, 1.5)\\
Age 30-55: (2.93, 4, 1.86, 1.23)\\
Age 56+: (2.99, 2.62, 2.42, 5.5)
\end{tabular} & Average contact rates by age category and household size \\
\hline
$I_0$ & \begin{tabular}[c]{@{}l@{}}
Age 0-15: (0, 20, 0, 60)\\
Age 16-29: (20, 20, 20, 20)\\
Age 30-55: (20, 40, 20, 20)\\
Age 56+: (20, 20, 20, 20)
\end{tabular} & Initial number of infected individuals  \\
\hline
$H_I^0$ & \begin{tabular}[c]{@{}l@{}}
Size 1: 60\\
Size 2: 120\\
Size 3: 60\\
Size 4: 120
\end{tabular} & Initial number of infected households of each size\\
\hline
\end{tabular}
\label{tab:init}
\end{table}
\newpage

To simulate the system (\ref{sdeh}), we utilize the \texttt{DifferentialEquations} package in \texttt{Julia}, implementing the \texttt{ISSEulerHeun()} method, an order 0.5 split-step Stratonovich implicit solver. The choice of an implicit solver is motivated by the system's numerical stiffness, which arises from the different magnitudes of coefficients in the population and household equations (see e.g. Chapter 8 in \citeauthor{Sarkka2019}, \citeyear{Sarkka2019}).

\subsection{Comparison between the Model With Household Structure and Other Models}\label{sec:comparison_othermodels}
To systematically compare different model structures, we examine four model variants:

\begin{itemize}
    \item Model U: Unstructured population (L=1, A=1)
    \item Model A: Age-stratified only (L=1, A=4)
    \item Model H: Household size-stratified only (L=4, A=1)
    \item Model AH: Our full age-household structure model from Section~\ref{sec: model} (L=4, A=4)
\end{itemize}

Detailed formulations of models U, A, and H are provided in Appendix \ref{secA}. To ensure a meaningful simulation-based comparison across these structurally different models, we adopt a reverse engineering approach: we first fix the desired final epidemic size (i.e.the cumulative proportion of the population that has been infected by the end of the epidemic) at 80\% for all models, then numerically determine the transmission ($\alpha$) and recovery ($\beta$) rate combinations that achieve this target. This approach provides a common baseline for comparison, as the final epidemic size serves as an observable, meaningful metric \citep{pellis2020systematic}, while parameters like transmission rates carry different interpretations across models.

To focus specifically on model structure rather than demographic variations, we consider the same recovery rates value ($\beta$) across all subpopulations within each stratified model and use a balanced population size distribution. For the parameter search, we employ a Sobol' sequence-based method \citep{Sobol1967OnTD} to efficiently identify the combinations of $\alpha$ and $\beta$ that yield our target 80\% final epidemic size. This quasi-random sampling approach systematically explores the parameter space more effectively than traditional grid search methods \citep{Kucherenko2015ExploringMS}, which is valuable as the relationship between parameters and final size becomes more complex in structured models.

Since the parameters $\alpha$ and $\beta$ carry fundamentally different interpretations across models due to their varying structures, direct comparison of these parameters would not provide meaningful insights. Instead, we analyze the temporal dynamics of the epidemics through the distribution of time to peak infection. For each model, we simulate the epidemic using the parameter combinations ($\alpha$, $\beta$) identified from our backward search, allowing us to compare the models' behavior while maintaining a consistent outcome measure.

In the following analysis, we examine the distribution of time-to-peak-infection across the four model variants. Using violin plots on a logarithmic scale (Figure \ref{fig:peak time violins}), we can visualize how different model structures influence both the average timing of the epidemic peak and its variability. 

\begin{figure}[H]
    \centering
    \includegraphics[width=\textwidth]{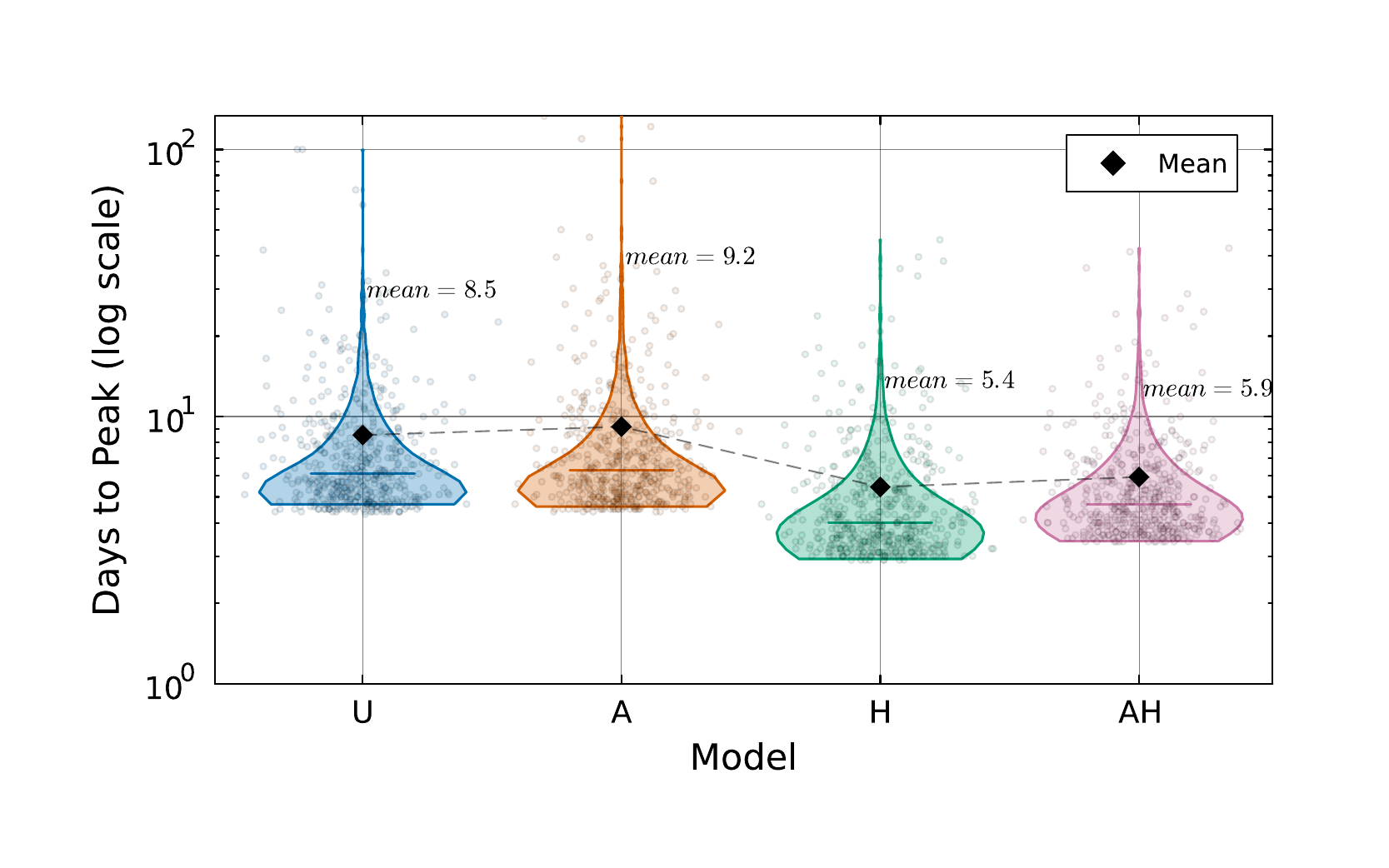}
    \caption{Comparison of time-to-peak distributions across four epidemic models (U: Unstructured, A: Age only, H: Household only, and AH: Age and household).  Each violin displays the distribution of peak times from $500$ unique parameter combinations $(\alpha,\beta)$, all of which were constrained to produce an $80\%$ final epidemic size. The distributions are shown on a logarithmic scale (base 10). Black diamonds represent the mean value for each model, while  individual data points are overlaid in gray to show the raw distribution. The household-based models (H, mean 5.4 days; AH, mean 5.9 days) peak earlier than the non-household models (U, mean 8.5 days; A, mean 9.2 days).}
    \label{fig:peak time violins}
\end{figure}
The comparison in the plot reveals a clear distinction between models with and without household structure. Models that incorporate household dynamics (H and AH) consistently peak earlier (means of 5.4 and 5.2 days, respectively) compared to their non-household counterparts (U and A, with means of 8.5 and 9.2 days).  To statistically validate these visual observations, we first analyzed the underlying data: the $500$-sample peak time vectors for each of the four models. As suggested by the asymmetric, long-tailed shapes in the plot, formal normality tests (e.g. Shapiro-Wilk) confirmed that all four distributions are significantly skewed and non-normal. This non-normality, combined with our $2\times 2$ factorial model design (Age: Yes/No, Household: Yes/No), makes a standard two-way ANOVA inappropriate. We, therefore, employed the Schneirer-Ray-Hare test, which is the non-parametric equivalent. This test allows us to move beyond a simple four-group comparison and instead quantify the independent main effect of each structural component (age and household size) as well as any potential interaction effect between them.  The statistical results from the Scheirer-Ray-Hare test \citep{sokal1995biometry} 

are summarized in Table~\ref{tab:srh_results}.

\begin{table}[htbp]
\centering
\caption{Results of the Scheirer-Ray-Hare test on the time to peak. This test analyzes the $2\times 2$ factorial design (age $\times$ household size) on the ranks of the peak-time data. Df: degrees of freedom, $\chi^2$: test statistic}
\label{tab:srh_results}
\begin{tabular}{lccr}
\toprule
Source of variation & Df & $\chi^2$ & $p$-value \\
\midrule
\textbf{Household} (main effect) & 1 & 537.11 & \textbf{$< 0.001$} \\
\textbf{Age} (main effect) & 1 & 0.00 & 0.947 \\
\textbf{Household $\times$ Age} (interaction) & 1 & 0.74 & 0.389 \\
\bottomrule
\end{tabular}
\end{table}

As shown in Table~\ref{tab:srh_results}, 
the addition of household structure has a 
statistically significant impact on accelerating the time to peak ($p < 0.001$). This acceleration is likely due to the dual transmission pathways in household models. In contrast, the addition of age structure alone does not show a statistically significant main effect on peak timing ($p = 0.947$). This is consistent with our visual observation of only minor mean shifts when adding age (e.g., U to A). Finally, the test found no significant interaction effect ($p = 0.389$), indicating the large effect of households is consistent regardless of whether age is also included. These findings allow us to conclude that household structure is the dominant factor influencing the temporal dynamics of disease spread in these models. These results have clear implications for public health, as they indicate that model structure, specifically the inclusion of household dynamics substantially affects the predicted timing of peak infection, which would influence the timeline for implementing intervention strategies.

\subsection{Decomposition of Transmission Pathways: Public versus Household Transmission}\label{sec:transmission_pathways}
Having established the different peak timing behaviors across model structures, we now focus specifically on our household-structured model (AH) to examine the underlying mechanisms driving these differences by decomposing the transmission pathways in our household-structured model. Specifically, we distinguish between public transmission and household transmission within our model to understand their relative contributions to disease spread.
The model incorporates two primary transmission pathways: public and household transmission. For a given household size category and a specific age category, constituting subpopulation~$k\in\{1,\ldots,K\}$, the public transmission force is given by:
\begin{equation*}
\alpha   s_k \sum_{j=1}^K c_{kj}^p i_j.
\end{equation*}

\noindent The household transmission force can be derived from the drift term of~$s_k$, i.\,e., $\mu_k^S$ as displayed in Equation~\eqref{SDE_terms}, when restricting the contact rate~$c_{kj}$ from Equation~\eqref{formula_ckj} to the household contribution. It results as:
\begin{equation*}
    \alpha s_k \Big[\sum_{l=1}^L z_{k,l}h_I^lH_{l}\Big]\Big[\sum_{l=1}^L z_{k,l}(1-h_I^l)^{-1}\Big]\frac{1}{N_H^k}\sum_{j=1}^K c_{kj}^hi_j.
\end{equation*}
To quantify the relative importance of each pathway, we aggregate the forces across all subpopulations. At time $t$, the total public transmission force is:

$$F^p(t) = \sum_{k=1}^K\sum_{j=1}^K \alpha s_k(t) c_{kj}^p i_j(t),$$
and the total household transmission force is:

$$F^h(t) = \sum_{k=1}^K\sum_{j=1}^K \alpha s_k(t) \Bigl[\sum_{l=1}^L z_{k,l}h_I^l(t)H_{l}\Bigr]\Bigl[\sum_{l=1}^L z_{k,l}(1-h_I^l(t))^{-1}\Bigr]\frac{1}{N_H^k}c^h_{kj}i_j(t).$$
Let $t^*$ denote the peak time where $F^p(t) + F^h(t)$ reaches its maximum.
The relative contributions of each pathway are calculated as:
\begin{equation*}
    \frac{F^p(t^*)}{F^p(t^*) + F^h(t^*)} \qquad\text{ (public contribution)}
\end{equation*}
and
\begin{equation*}
    \frac{F^h(t^*)}{F^p(t^*) + F^h(t^*)}  \qquad\text{ (household contribution).}
\end{equation*}

\begin{figure}[hbtp]
    \centering
    \includegraphics[width=1\linewidth]{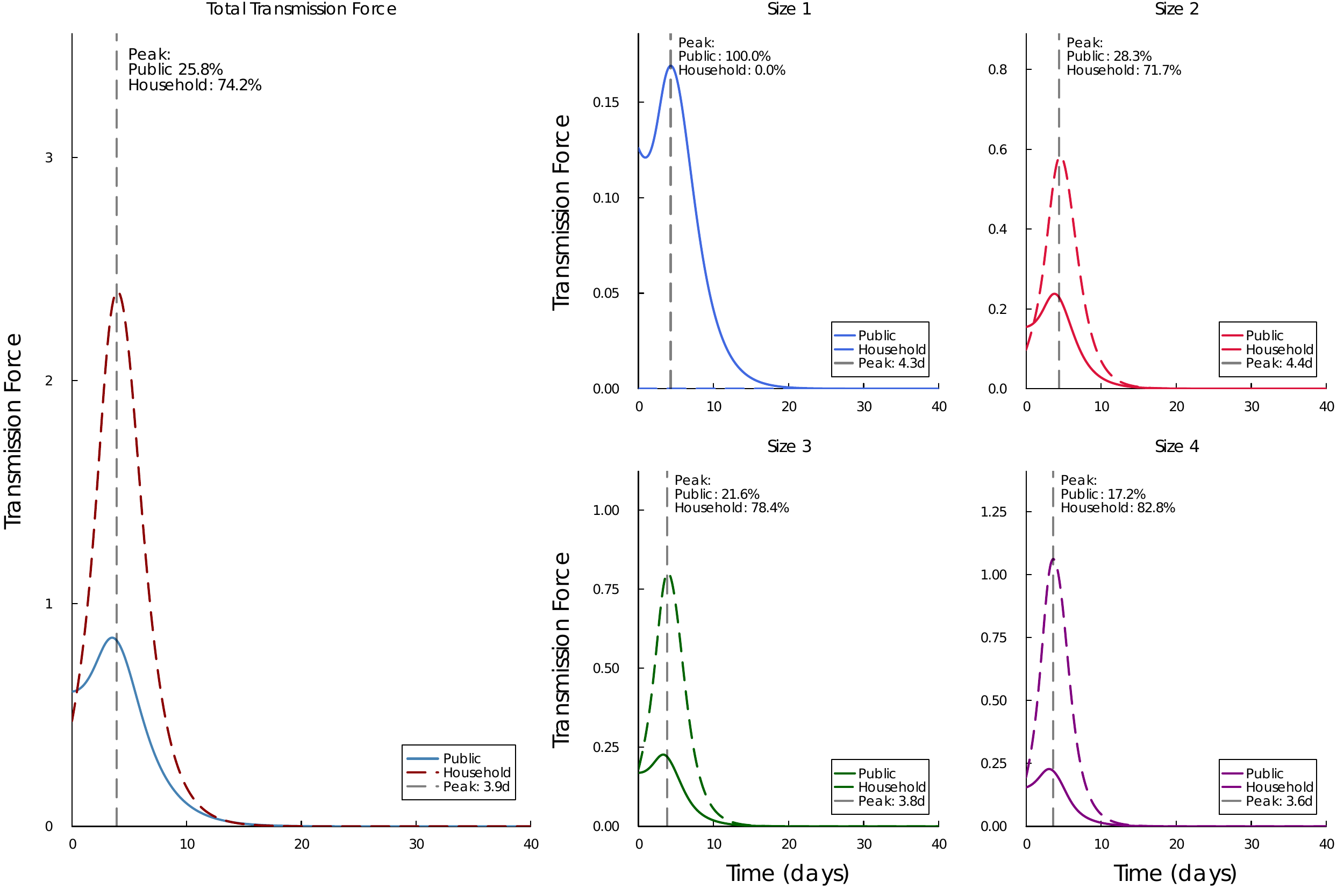}
    \caption{Decomposition of public ($F^p$, solid) and household ($F^h$, dashed) transmission forces over time (final size 80\%). Left: total transmission across all subpopulations. Right: transmission dynamics by household size. Vertical dashed lines mark epidemic peaks. Percentages show relative contributions at peak, demonstrating that household transmission dominates in larger households while public transmission prevails in smaller households.}
    \label{fig:transmission_pathways}
\end{figure}

Decomposition of public ($F^p$, solid) and household ($F^h$, dashed) transmission forces over time (final size 80\%). Left: total transmission across all subpopulations. Right: transmission dynamics by household size. Vertical dashed lines mark epidemic peaks. Percentages show relative contributions at peak, demonstrating that household transmission dominates in larger households while public transmission prevails in smaller households.

Our analysis reveals distinct transmission patterns across household sizes (see Figure \ref{fig:transmission_pathways}). In the total population, transmission force peaks early (at day 3.9), with household transmission accounting for 74.2\% and public transmission 25.8\% of the total force at peak time. However, this aggregate view masks heterogeneity across household sizes. For single-person households (size 1), transmission is naturally entirely through public contacts (100\% of force). In contrast, larger households show predominant household transmission, highlighting the increasing importance of within-household spread as household size grows.

\subsection{Effect of the Public Policy on the Exposure Intensity of the Households}\label{sec:effect_public_policy}
We now examine how public policy interventions affect the dynamics of our stochastic epidemic model, an important element of mathematical epidemiology. Specifically, we analyze how policy-driven changes in transmission patterns influence household exposure risks in the presence of inherent stochasticity in disease transmission processes and variability in policy implementation. 

To quantify how the proportion of infected households affects the exposure risk for susceptible households, we use the size-specific exposure intensity function introduced in \cite{Bayham2016}. For households from size category $l$, this function is defined as:
\begin{equation*}
m_l(t) = \frac{1}{1 -h_{I,l}(t)},
\end{equation*}
where $h_{I,l}(t)$ represents the proportion of infected households from size category~$l$ at time $t$. This function serves as an amplification factor in our model as already introduced in the context of Equation~\eqref{contact_intensity_multiplier}, capturing the increased risk of infection for susceptible households as more households of the same size become infected: When few households are infected ($h_{I,l}(t)$ close to zero), the exposure intensity is near one. As the proportion of infected households increases, the exposure intensity grows nonlinearly, reflecting heightened transmission risk when infections become widespread in the community. This formulation enables us to assess how variations in household size influence the exposure dynamics over time and evaluate the effectiveness of size-specific public health interventions.

We aim to examine the effects of varying the public policy parameters~$\omega_k$ for~$k=1,\dots,K$ on the exposure intensity across households of different sizes. To achieve this, we explore two distinct scenarios: a stringent lockdown of households of size 4 (i.\,e.\  $\omega_{13}=\ldots=\omega_{16}=0.1$, else~$\omega_k=1$), and an unrestricted scenario ($\omega_k=1$ for all~$k=1,\ldots,16$). While this distinction between sizes~3 and~4 is implausible in real life, it serves as a model of realistic interventions in larger shared households. The framework allows for targeted interventions, enabling the differentiation of public policy measures among specific subpopulations.

\begin{figure}[H]
    \centering
    \includegraphics[width=\textwidth]{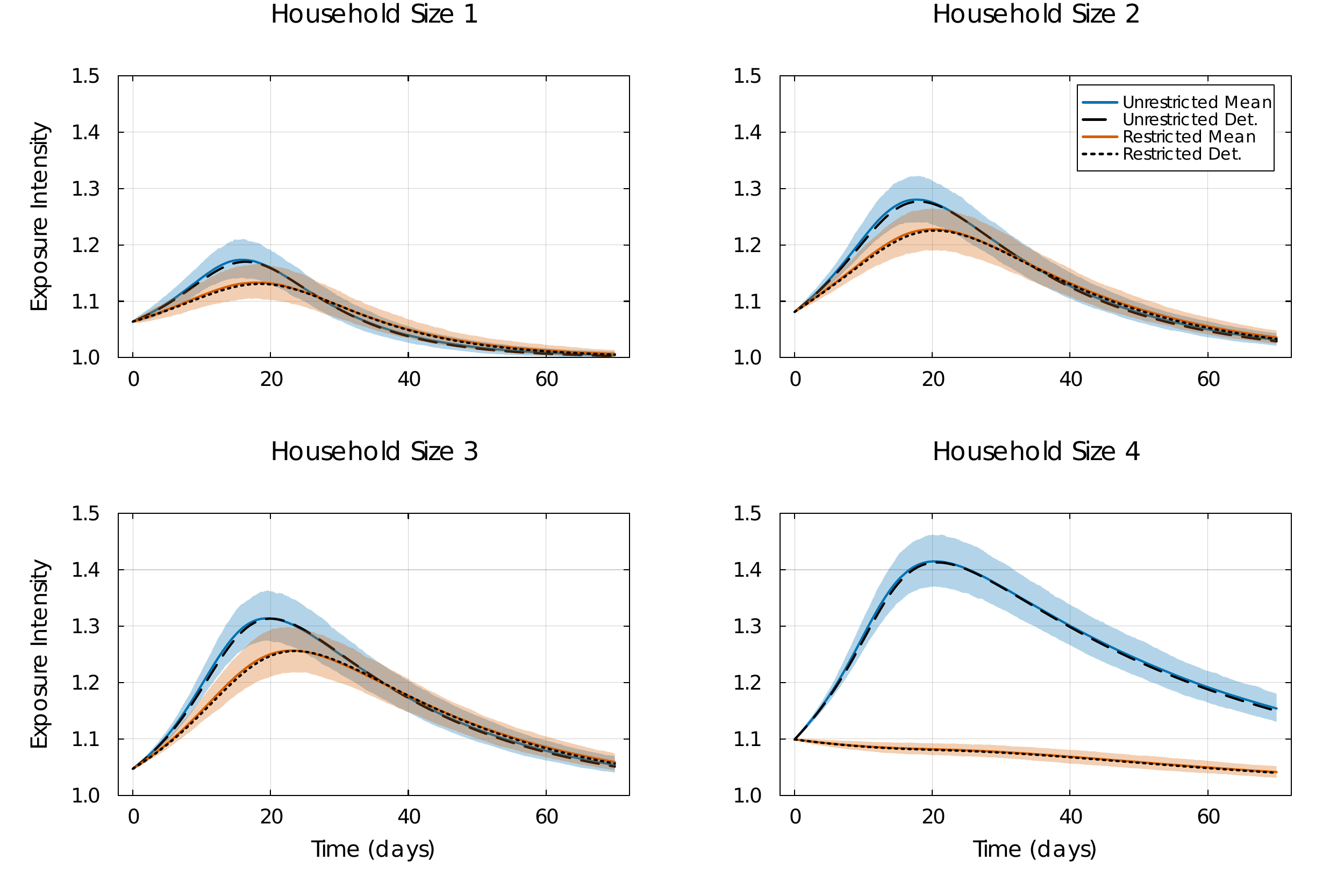}
    \caption{Exposure intensity dynamics for unrestricted (no lockdown) and restricted (partial lockdown) scenarios
    for households of different sizes with 95\% ranges (shaded areas). The restricted scenario implements 
    public contact reduction ($\omega_k=0.1$) for individuals in size-4 households.}
    \label{fig:exposure_intensity}
\end{figure}

The analysis of exposure intensity dynamics reveals both direct and indirect effects of public contact restrictions across household sizes. While restrictions were specifically applied to households of one size category (size~4), their impact propagates throughout the system, resulting in reduced peak exposure intensities across all household categories. For households of size 1, where individuals live alone and thus have no household contacts, the exposure intensity is solely determined by public interactions, explaining their distinctly lower exposure intensity compared to other household sizes. Larger households (sizes 2 and 3) exhibit progressively higher peak exposure intensities in both scenarios compared to households of size $1$, reflecting the combined effect of both household and public transmission pathways. The temporal evolution follows a consistent pattern across household sizes~2 and ~3: an initial increase leading to a peak around day~16 (unrestricted scenario) or day 20 (restricted scenario), followed by a gradual decline. Households of size~4, which are directly targeted by the public contact restrictions, display the most pronounced difference between the two scenarios. In the unrestricted case, these households reach the highest peak exposure intensity among all household sizes, while under restrictions, their peak is substantially reduced and delayed, demonstrating the direct effectiveness of the intervention on the targeted group. Notably, the restricted scenario not only reduces peak intensities but also delays their timing, with the effect being more pronounced in larger households. The stochastic simulations, represented by the mean (solid lines) and 95\% ranges (shaded areas), closely align with the deterministic solutions (black lines), particularly in smaller households. However, larger households display wider 95\% ranges towards the end, indicating greater variability in their transmission dynamics. This comprehensive analysis demonstrates that implementing restrictions on public contacts in larger households can have system-wide effects on exposure intensity, highlighting the interconnected nature of disease transmission across different household structures.

\subsection{Global Sensitivity Analysis of the Threshold Parameter}\label{sec:GSA}

From the approximation of the underlying stochastic process by branching processes near the disease-free equilibrium, we derived the basic reproduction number $\mathcal{R}_0$ in Section~\ref{sec:branching_process_approx}. The outbreak of the disease is determined by $\mathcal{R}_0$ as outlines in Section~\ref{sec:branching_process_approx}. 
Given the expression of $\mathcal{R}_0$ and its critical role in determining the system's behavior, it is essential to understand which parameters most significantly influence its value.

To identify the most critical parameters influencing $\mathcal{R}_0$ and consequently influencing policy interventions, we employ Sobol's variance-based sensitivity analysis \citep{sobol2001global}, a prominent global sensitivity analysis (GSA) method used to study how uncertainty in a model's output can be apportioned to its different input variables. This method quantifies the relative importance of input parameters by decomposing the output variance into contributions from individual parameters and their interactions. We focus on two specific measures: first-order Sobol indices $S_1$, which capture the direct effect of a parameter, and total Sobol indices $S_{T}$, which measure its total effect including all interactions. Both indices range from zero to one, where a value close to one indicates strong parameter influence on the output, while a value close to zero indicates negligible influence. The difference between~$S_1$ and~$S_{T}$ reveals the importance of a parameter's interactions with other variables. For this study, we apply the GSA to a policy-relevant binary outcome: the probability of an epidemic occurring (i.\,e., the probability of~$\mathcal{R}_0 > 1$). The analysis is specifically focused on the policy parameters $\omega_k$, one for each subpopulation~$k=2,\ldots,16$, while all other model parameters are held constant as defined earlier. (Subpopulation~1, i.\,e., children in single households, is assumed empty.)

The sensitivity of the epidemic threshold to each policy parameter was quantified using first-order $S_1$ and total-effect $S_T$ Sobol indices, calculated from evaluations of the expression of $\mathcal{R}_0$ across $10^5$ combinations of parameters. The results, summarized in Figure~\ref{fig:sobol_indices}, reveal that the model's outcome is overwhelmingly dominated by the policy parameter corresponding to a single subpopulation: the policy applied to seniors (56+) in households of size 4 has a first-order index~$S_1$ of~$0.349$, which is an order of magnitude larger than any other parameter. This indicates that the direct effect of this single policy accounts for approximately $35\%$ of the total variance in whether an epidemic occurs.

Furthermore, the total-effect indices and the summary of interactions highlight the pivotal role of this particular subpopulation: the $S_T$ index for its policy parameter is~$0.823$, indicating that this parameter is involved in nearly all influential interactions. This suggests that the effectiveness of policies on other subpopulations is largely dependent on their interaction with this single, highly influential group.

\begin{figure}[H]
    \centerline{
    \includegraphics[width=1.1\textwidth]{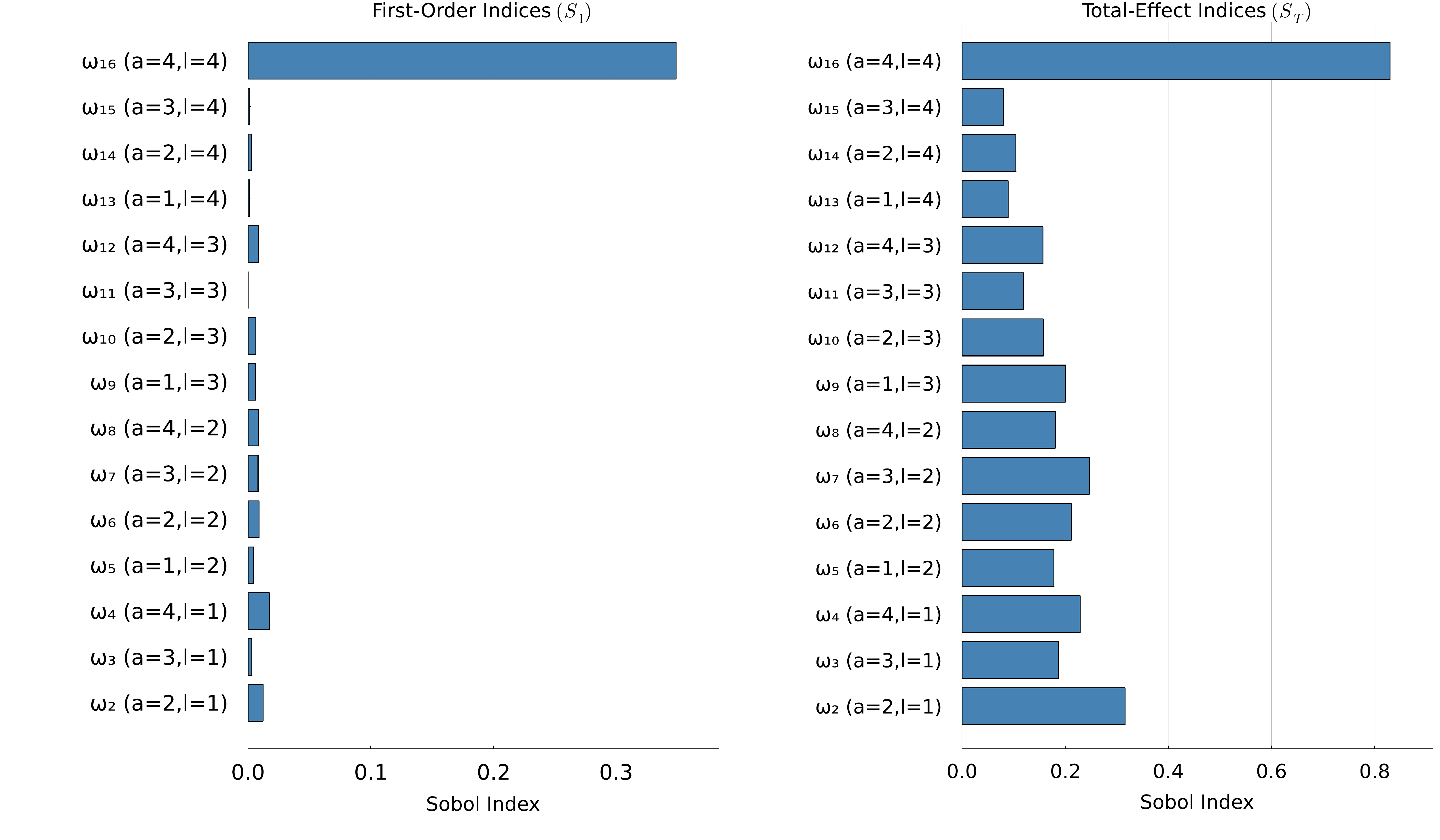}}
    \caption{First-order ($S_1$, left) and total-effect ($S_T$, right) Sobol indices for the policy parameters $\omega_2,\ldots,\omega_{16}$ on the probability that $\mathcal{R}_0 > 1$. The policy targeting seniors (age category~$a=4$) in households of size 4 (size category~$l=4$), i.\,e.\, subpobulation~16,  is shown to be the dominant factor.}
    \label{fig:sobol_indices}
\end{figure}

The affected subpopulation~16 stands out from the others in that it exhibits the highest activity level~$a_{16}^p$ (see Table~\ref{tab:init}), followed by subpopulation~2. To investigate the cause of high sensitivity further, we related a subpopulation's public activity level~$a_k^p$ to the total-effect index~$S_T$ of its corresponding policy parameter~$\omega_k$. As illustrated in Figure~\ref{fig:correlation}, we found a positive correlation between these two variables, with a calculated Pearson correlation coefficient of $\rho \approx 0.843$.

\begin{figure}[H]
    \centering
    \includegraphics[width=0.9\textwidth]{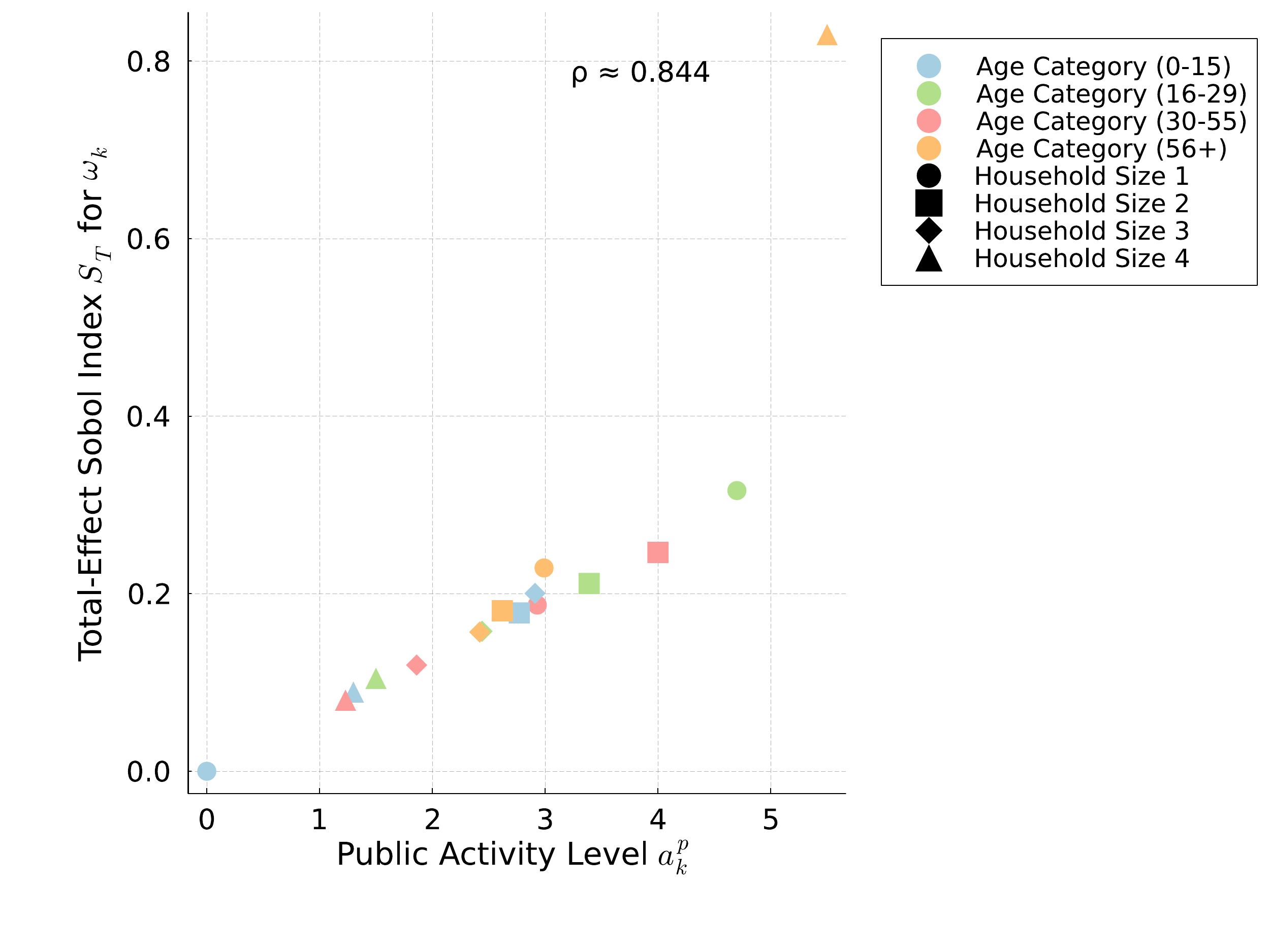}
    \caption{Relation between the public activity level $a^p_k$ of each of the 15 subpopulations and the total-effect Sobol index ($S_T$) of its corresponding policy parameter $\omega_k$. Each point represents a subpopulation, colored by age category and shaped by household size. The Pearson correlation coefficient equals~$\rho\approx 0.844$.}
    \label{fig:correlation}
\end{figure}

The GSA suggests that an epidemic in this model is not driven uniformly by all subpopulations but is instead contingent on the behavior of a single "linchpin" group. The policy targeting that subpopulation was identified as the critical intervention point, a direct result of this group having the highest public activity level in the model's parameterization.

The strong correlation between a group's activity level and its policy importance confirms that a parameter's sensitivity is not arbitrary but is a direct consequence of its structural role in the model's contact dynamics. This quantitatively validates the intuition that restricting the most socially active groups yields the greatest impact on disease transmission.

From a public health perspective, these results advocate for a targeted, data-driven intervention strategy over broad, undifferentiated lockdowns. The analysis suggests that identifying and focusing resources and restrictions on the subpopulation with the highest public contact rates would be the most efficient and effective means of controlling an outbreak. By mitigating transmission within this core group, the efficacy of all other transmission pathways, which are shown to be highly dependent on interactions with this group, is also reduced. This modeling outcome underscores the critical need for real-world behavioral data to inform and prioritize public health policy.

\subsection{Effect of Stochasticity}\label{sec: stoch_effect}

The integration of household dynamics into stochastic models, as described in Section~\ref{sec: model}, represents an innovation of our work compared to more common deterministic models. In the following, we examine the effect of stochasticity in our household-structured epidemic model by comparing simulations from SDE and ODE models. The SDE model has been described by Equation~\eqref{sdeh} with drift and diffusion terms given by Equation~\eqref{SDE_terms}. For the ODE model, we set all diffusion terms to zero. Our comparison explores how incorporating stochastic terms affects not only the epidemic trajectories but also the timing of peak infections across different household sizes and age categories. To this end, we simulate 1000 independent trajectories from the SDE model and explore pointwise means and variances over time. For the ODE model, there is only one simulation, with no measurement error assumed.

Figure \ref{fig:infected_age_dynamics} displays simulated proportions~$i_k$ of infected individuals in the subpopulations~$k=2,\ldots,16$. It shows that the deterministic behaviour of the ODE model and the mean behaviour of the SDE model are similar within each subpopulation; the SDE model, however, captures natural fluctuations around the mean trajectories, shown by the shaded empirical 95\% ranges around the solid lines. These stochastic effects are particularly pronounced during peak infection periods around, i.\,e., between day~15 and~20. 

\begin{figure}[H]
    \centering   
    \includegraphics[width=\textwidth]{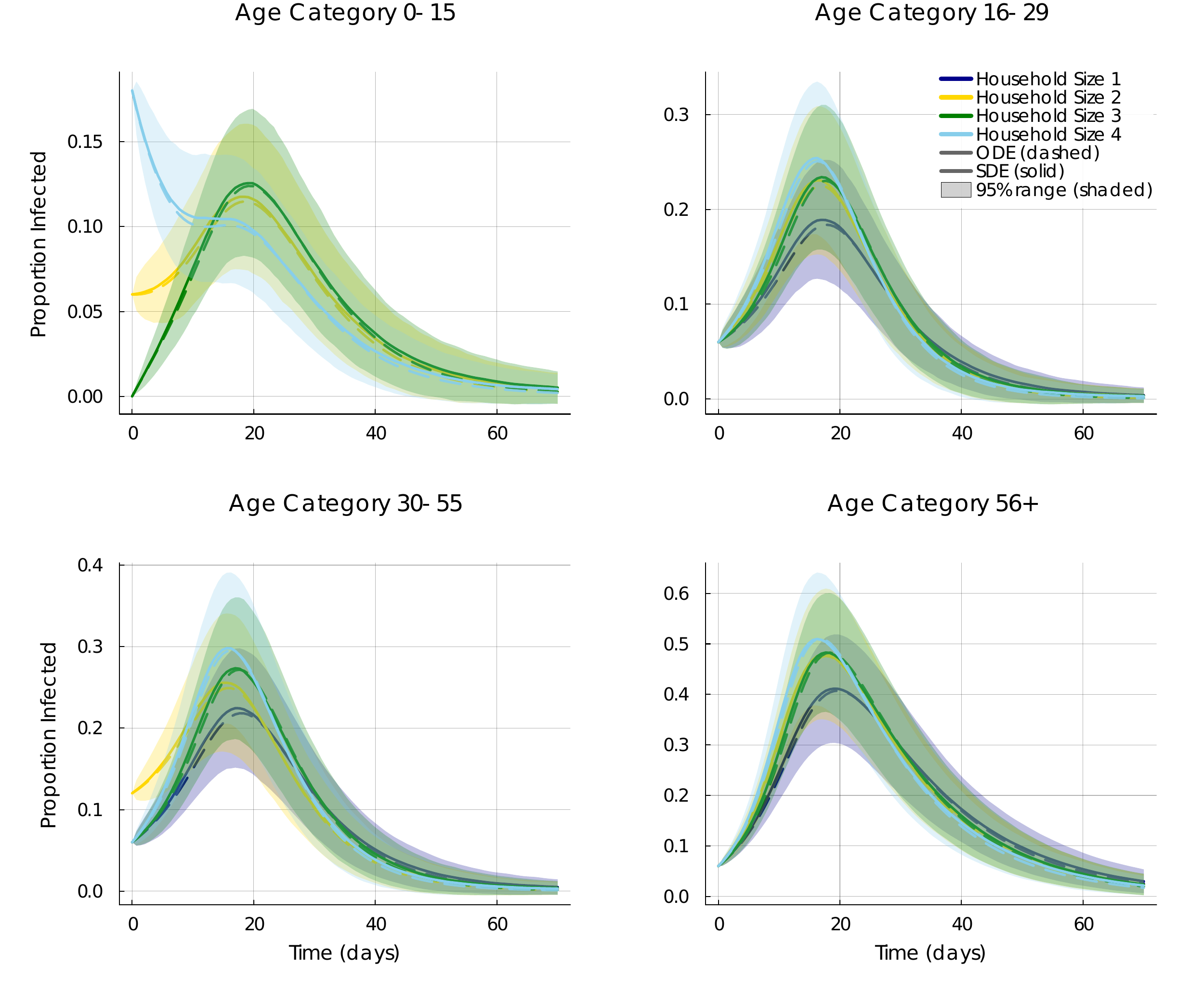}
    \caption{Temporal dynamics of infected proportions~$i_k$ in subpopulations~$k=2,\ldots,16$. Each subplot represents a different age category and displays courses for different household size categories. The plots show deterministic (dashed lines) and stochastic (solid lines) solutions. The shaded regions around the solid lines indicate empirical pointwise 95\% ranges around the mean of SDE trajectories across 1000 simulations, capturing the inherent variability in the infection dynamics.}
    \label{fig:infected_age_dynamics}
\end{figure}

Moving the focus from the individual-level to the household-level dynamics~$h_{S,l}$ and~$h_{I,l}$, $l=1,\ldots,4$, leads to a similar observation pattern: as visible from Figure \ref{fig:household_dynamics}, the deterministic ODE solutions are close to the means of stochastic trajectories for both susceptible and infected household proportions; the variability of the SDE model is expressed through empirical 95\% ranges.

\begin{figure}[H]
    \centering    
    \includegraphics[width=\textwidth]{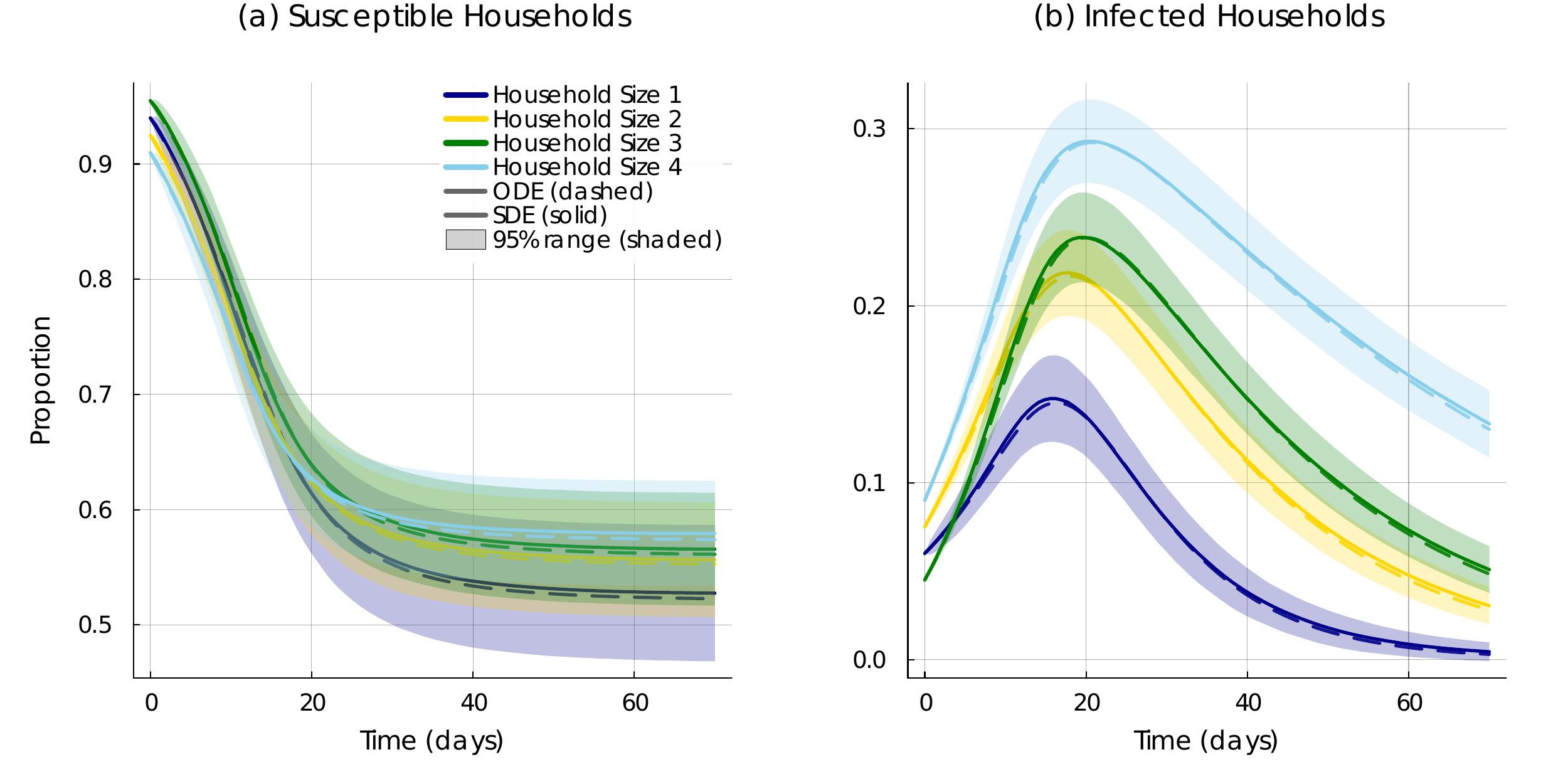}
    \caption{Temporal dynamics of household states~$h_{S,l}$ (left) and~$h_{I,l}$ (right), $l=1,\ldots,4$. The plots show deterministic (dashed lines) and stochastic (solid lines) solutions. The shaded regions around the solid lines indicate empirical pointwise 95\% ranges around the mean of SDE trajectories across 1000 simulations.}
    \label{fig:household_dynamics}
\end{figure}

To quantify the magnitude of stochastic effects which became visible in Figures~\ref{fig:infected_age_dynamics} and~\ref{fig:household_dynamics}, we consider the coefficient of variation (CV). For each time point~$t$, it is calculated as $\text{CV}(t) = s(t)/m(t)$, where $s(t)$ is the empirical standard deviation and $m(t)$ the empirical mean across the~1000 stochastic realizations at time~$t$. This time-dependent CV provides a standardized measure of the relative variability in our ensemble of trajectories, enabling us to quantify how stochastic effects vary throughout the epidemic progression and compare variability across different subpopulations.

Figure \ref{fig:cv_analysis} reveals patterns in the relative variability of the epidemic dynamics. At the individual level (top panel), CV values cluster primarily by household size. Households of size~4 (light blue) exhibit the highest CV values, reaching approximately~1.2, while smaller household sizes display progressively lower variability. Within each household size, there is some age-specific variation in CV patterns during the epidemic peak (days 15-25). Notably, one curve shows a substantially delayed onset (beginning around day 10-15), corresponding to a subpopulation with zero initial infections where variability only emerges once infections spread to that group.

At the household level, contrasting patterns emerge between susceptible (middle panel) and infected (bottom panel) households. For susceptible households, all sizes maintain relatively low CV values (below 0.05), with single-person households (dark blue) showing slightly higher variability than larger households. For infected households, the pattern is more pronounced: single-person households exhibit substantially higher CV values (reaching approximately 0.6) compared to larger households, which maintain CV values below 0.2. These contrasting patterns highlight the practical value of incorporating stochasticity into household-structured epidemic models. The SDE framework reveals that epidemic uncertainty differs fundamentally by scale: larger households (size~4) show high variability in infection proportions (top panel) due to amplified within-household transmission, while single-person households show high variability in infection status (bottom panel) due to dependence on external contacts. Capturing this scale-dependent variability—which deterministic models miss entirely—enables more realistic uncertainty quantification for epidemic forecasting and intervention planning.

\begin{figure}[tbhp]
    \centering
    \includegraphics[width=\textwidth]{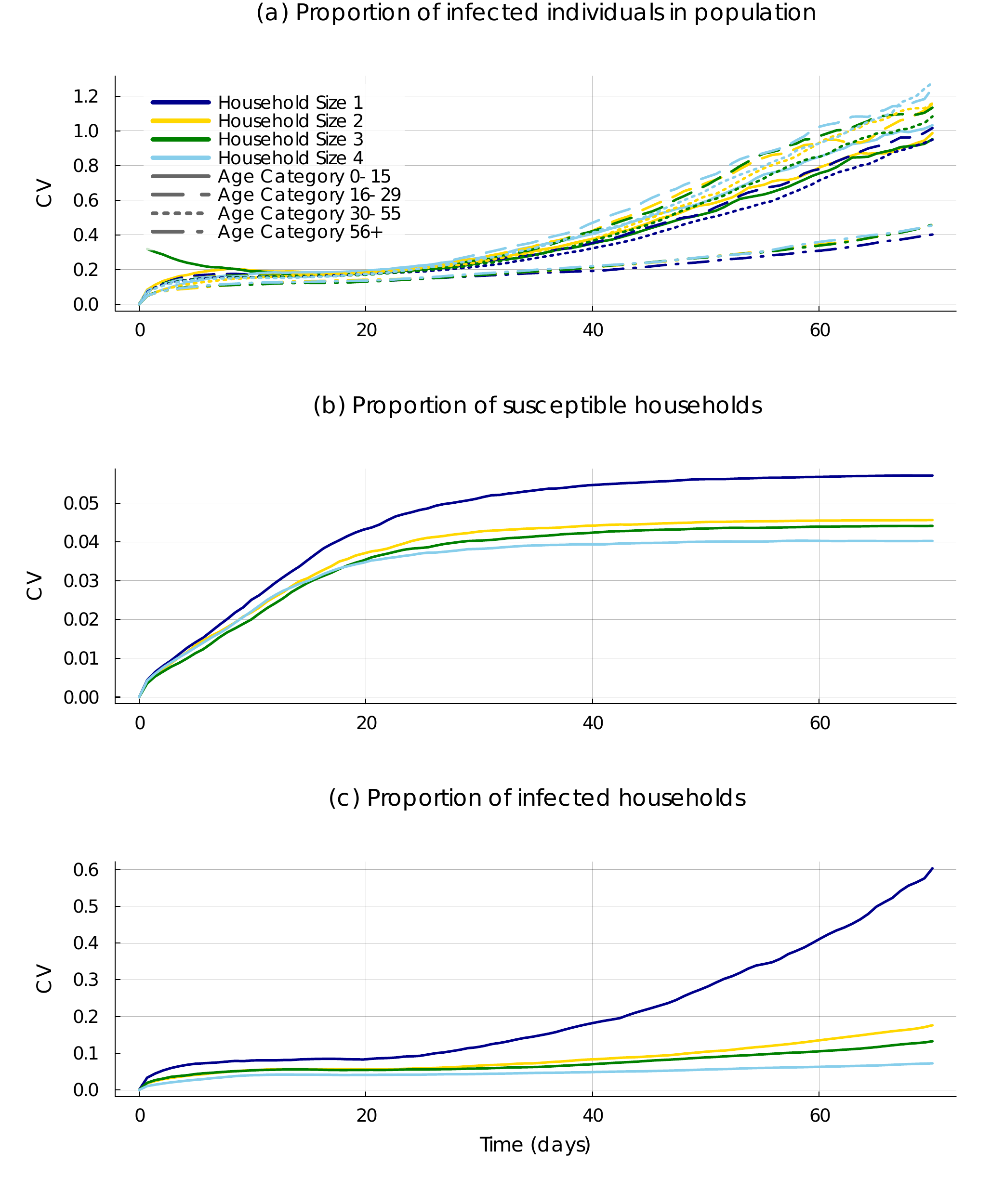}
    \caption{Coefficient of variation (CV) analysis showing the relative variability in the system, based on 1000 simulations: (Top) CV for proportion of infected individuals~$i_k$ across subpopulations. The curves apparently cluster by household size. (Middle) CV for porportion of susceptible households~$h_{S,l}$ and (Bottom) infected households~$h_{I,l}$.
    }
    \label{fig:cv_analysis}
\end{figure}
Figure \ref{fig:peak_time_analysis} examines the distribution of epidemic peak timing across subpopulations. The total population exhibits a relatively narrow timing distribution centered at day~17.1, compared to the wider distributions observed in individual subpopulations, demonstrating that population-level averages mask substantial temporal heterogeneity that only stochastic models can capture. Across most household sizes, the youngest age category (0-15) displays notably higher variability in peak timing compared to older age groups, whose distributions are more concentrated, revealing that stochastic effects are not uniform but depend critically on demographic characteristics. The (0-15) age category in size~4 households represents an exceptional case with a mean near-zero peak time ($m = 0.2$ days): the high initial infection proportion in this subpopulation causes infections to decline from $t=0$, meaning the initial condition itself represents the peak rather than an epidemic-driven peak—a phenomenon that highlights how stochasticity interacts with initial conditions to create qualitatively different epidemic trajectories. For other subpopulations in the (0-15) age category, the greater variability in peak timing may be related to their higher recovery rate ($\beta = 0.2$, see Table~\ref{tab:init}) combined with their lower contact rates and more heterogeneous initial conditions (including zero initial infections in some subpopulations), which together make their infection dynamics more dependent on stochastic transmission from other age groups. In contrast, the 56+ age category, with the lowest recovery rate ($\beta = 0.071$), shows more predictable infection accumulation and consequently more consistent peak timing, illustrating that certain demographic groups are inherently more susceptible to stochastic variability than others. Collectively, these patterns demonstrate that the SDE framework is essential for capturing the full spectrum of epidemic uncertainty across heterogeneous populations, from timing variability to demographic-specific sensitivities.

\begin{figure}[tbhp]
    \centering
    \includegraphics[width=\textwidth]{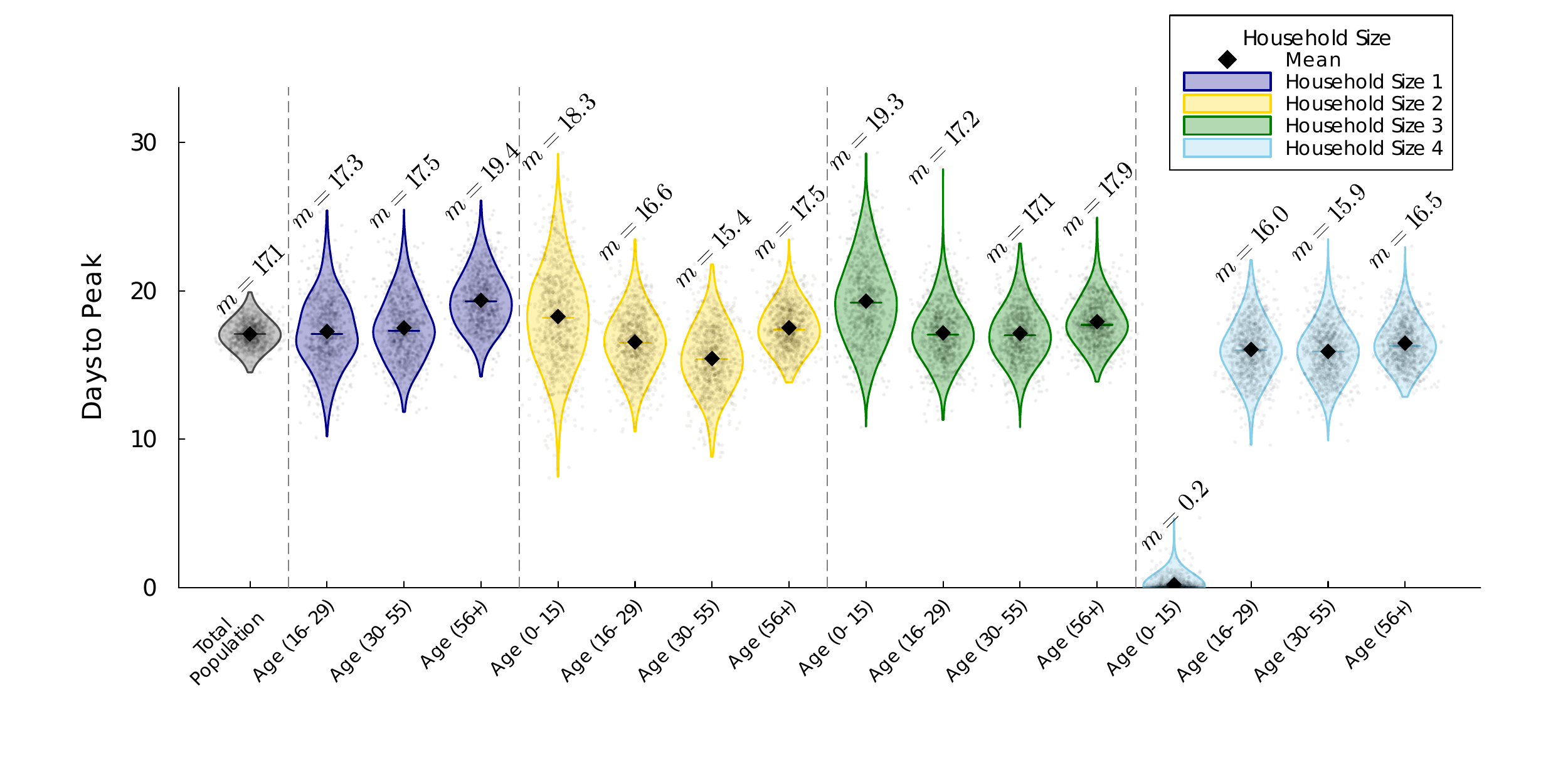}
    \caption{Distribution of time to peak infection for total population and across subpopulations. Violin plots show the empirical distribution of peak timing over 1000 simulations, with width indicating frequency at each time point. Black dots represent individual trajectories, and larger gray points mark the mean ($m$) for each subpopulation, whose value is also shown above each violin. The plot shows the result for the aggregated total population (gray) at the left, followed by subpopulations organized by household size  and age categories within each size group.  
}
    \label{fig:peak_time_analysis}
\end{figure}

\section{Conclusion} \label{sec:conclusion}
We introduced an SDE model that incorporates household structure as a central feature of the SIR framework, distinguishing between within-household and public transmission pathways. This approach aims to address limitations of existing models by capturing the randomness inherent in disease spread while explicitly representing how household composition influences infection dynamics.

A key feature of our model is its flexible approach to contact rates between subpopulations. This flexibility allows for a broader spectrum of social behaviors to be represented, which we believe is crucial for modeling scenarios where social interactions vary significantly across different demographics. The inclusion of a public policy parameter enables us to explore how lockdown measures and restrictions on external contacts might influence disease transmission. Given the recent global experiences with pandemic management, this feature could provide valuable insights for policymakers.

We derived the model's basic reproduction number using branching processes. A sensitivity analysis of this key figure revealed which factors most significantly influence disease spread in our model; that is, the policy parameter for the subpopulation with the highest activity level. These findings could help guide public health interventions, suggesting where efforts might be most effectively focused. Our exploration of quantitative exposure intensity offers a perspective on household vulnerability to infection in our stochastic model, which could be particularly relevant when considering targeted interventions or risk assessments.

The stochastic nature of our model allows for a more comprehensive representation of uncertainty in epidemic trajectories. This feature might be especially informative when modeling smaller populations or the early stages of an outbreak, where random fluctuations can have significant impact. Yet, it is worth noting that this added complexity comes at the cost of increased computational demands and potentially more challenging parameter estimation.

Several directions emerge for future research. Validation against empirical outbreak data would be an essential next step to assess the model's practical utility and comparative performance. Investigating how different population distributions and household structures impact model outputs could reveal important demographic considerations in epidemic spread. Extensions incorporating vital dynamics or waning immunity could enable analysis of endemic disease scenarios, expanding the model's applicability to long-term public health planning.

In conclusion, our SDE-based household-structured SIR model represents a new approach to epidemic modeling, offering new analytical capabilities while building upon established epidemiological frameworks. Through further development and refinement, it has the potential to enhance our understanding of disease transmission dynamics and inform public health strategies.
\section*{Statements and Declarations}

\paragraph{Competing Interests} The authors declare that they have no competing interests.

\paragraph{Funding} Our research was supported by the Helmholtz pilot project ”Uncertainty Quantification" and by the German Research Foundation (DFG) [RTG 2865/1 – 492988838].

\paragraph{Code Availability} The simulation code used in this study is available on GitHub at \citep{Yaqine_Stochastic_Differential_Equations_2025}. 

\paragraph{Acknowledgements}
The authors would like to express their gratitude to Dr. Sebastian Grube for his valuable discussions and insights that contributed to this work.
\begin{appendix}

\section{Appendix: Model Variations}\label{secA}

Section~\ref{sec:comparison_othermodels} numerically explores four different model structures, which are described in detail here.
\subsection{Model U: Unstructured SIR model, A=1, L=1}
\label{secA1}
In this model, there is no stratification by age or household size. The population is considered homogeneously mixed, and we use the standard SIR model without any subdivisions. \\[1ex]
\textbf{Variables}
\begin{itemize}
    \item[] $s(t)$: the proportion of susceptible individuals
    \item[] $i(t)$: the proportion of infected individuals
    \item[] $r(t)$: the proportion of recovered individuals
\end{itemize}
\textbf{Parameters}
\begin{itemize}
    \item[] $\alpha > 0$: the transmission rate
    \item[] $\beta > 0$: the recovery rate
\end{itemize}
\textbf{Equations}
\begin{itemize}
    \item[] Susceptible individuals: 
    \begin{equation*}
        ds(t)=-\alpha s(t) i(t) dt
    \end{equation*}
    \item[] Infected individuals:
    \begin{equation*}
        di(t)=\Big( \alpha s(t)i(t)-\beta i(t)\Big)dt
    \end{equation*}
    \item[] Recovered individuals:
    \begin{equation*}
        dr(t)=\beta i(t)dt
    \end{equation*}
\end{itemize}

\subsection{Model A: SIR model with stratification by age, A=4, L=1}\label{secA2}
In this model, the population is divided into four age categories. There is no stratification by household size. We assume that individuals mix uniformly within their age category.\\[1ex]
\textbf{Variables}
For each age category $k=1,\ldots,4$:
\begin{itemize}
    \item[] $s_k(t)$: the proportion of susceptible individuals in age category $k$
    \item[] $i_k(t)$: the proportion of infected individuals in age category $k$
    \item[] $r_k(t)$: the proportion of recovered individuals in age category $k$
\end{itemize}
\textbf{Parameters}\nopagebreak
\begin{itemize}
    \item[] $\alpha> 0$: the transmission rate
    \item[] $\beta_k> 0$: the recovery rate for age category $k$
    \item[] $c_{kj}> 0$: the contact rate between age categories $k$ and $j$   
\end{itemize}
\textbf{Equations}
\begin{itemize}
    \item[] Susceptible individuals:
    \begin{equation*}
        ds_k(t)=\Big(-\alpha s_k(t)\sum_{j=1}^4 c_{kj}i_j(t)\Big)dt
    \end{equation*}
    \item[]Infected individuals:
    \begin{equation*}
        di_k(t)=\Big(\alpha s_k(t)\sum_{j=1}^4 c_{kj}i_j(t)-\beta_ki_k(t)\Big)dt
    \end{equation*}
    \item[]Recovered individuals:
    \begin{equation*}
        dr_k(t)=\beta_k i_k(t)dt
    \end{equation*}
\end{itemize}
Susceptible individuals in age category $k$ become infected through contact with infected individuals from all age categories~$j=1,\ldots,4$, weighted by the contact rates $c_{kj}$. In this model, we have the total population in each age category $s_k(t)+i_k(t)+r_k(t)=1$ for all $t\geq 0$ and for each $k$.

\subsection{Model H: SIR model with stratification by household size, A=1, L=4}\label{secA3}
In this model, the population is divided based on household size categories. Age is not a differentiating factor. We assume here a homogeneous mixing within households sizes.\\[1ex]
    \textbf{Variables}
    For each household size $l=1,\ldots,4$:
         \begin{itemize}
        \item[] $s_l(t)$: the proportion of susceptible individuals in household size category $l$
        \item[]$i_l(t)$: the proportion of infected individuals in household size category $l$
        \item[]$r_l(t)$: the proportion of recovered individuals in household size category $l$
    \end{itemize}
   \begin{itemize}
       \item[]$h_{S,l}(t)$: the proportion of susceptible households of size category $l$
       \item[]$h_{I,l}(t)$: the proportion of infected households of size category $l$
       \item[]$h_{r,l}(t)$: the proportion of recovered households of size category $l$
   \end{itemize}
   \textbf{Parameters}
   \begin{itemize}
       \item[]$\alpha> 0$: the transmission rate
       \item[]$\beta> 0$: the recovery rate for individuals
       \item[]$\nu_l> 0$: the recovery rate for households of size category $l$
       \item[]$c_{lj}^p\geq 0$: the contact rate between individuals of household size category $l$ and $j$ in the public
       \item[]$c_{lj}^h\geq 0$: the contact rate between individuals of household size category $l$ and $j$ within households
       \item[]$N_l$: the total number of individuals in households of size category $l$
       \item[]$H_{l}$: the total number of households of size category $l$
       \item[]$N=\sum_{l=1}^4 N_l$: the total population size
       \item[]$H=\sum_{l=1}^4H_{l}$: the total number of households
   \end{itemize}
   \textbf{Equations} 
   \begin{itemize}
       \item[] Susceptible individuals
       \begin{equation*}
           ds_l(t)=\Big( \alpha s_l(t)\Big(  \sum_{j=1}^4 c_{lj}^p i_j(t)+H_{l} h_{I,l}\frac{c_{lj}^h}{N_{H,j}}\frac{h_{I,l}(t)}{1-h_{I,l}(t)}i_l(t)\Big)\Big) dt
       \end{equation*}
       \item[] Infected individuals
       \begin{equation*}
           di_l(t)=\left[\alpha s_l(t)  \sum_{j=1}^4 \Big(c_{lj}^p + \frac{h_{I,l}H_{l}}{N_{H,j}}\frac{h_{I,l}(t)}{1-h_{I,l}(t)}c_{lj}^h\Big)i_j(t)-\beta i_l(t)
           \right]dt
       \end{equation*}
       \item[]Recovered individuals
       \begin{equation*}
           dr_l(t)=\beta i_l(t)dt
       \end{equation*}
       \item[]Susceptible households
       \begin{equation*}
           dh_{S,l}(t)=\Big( -\alpha \frac{h_{S,l}}{H_{l}}(t)N_ls_l(t)\sum_{j=1}^4 c_{lj}^p i_l(t)\Big)dt
       \end{equation*}
       \item[]Infected households
       \begin{equation*}
           dh_{I,l}(t)=\left[  \alpha \frac{h_{S,l}}{H_{l}}(t)N_ls_l(t)\sum_{j=1}^4 c_{lj}^p i_l(t)-\nu_l h_{I,l}(t)
           \right]dt
        \end{equation*}
        \item[]Recovered households 
        \begin{equation*}
             dh_{R,l}(t)=\nu_l h_{I,l}(t)dt
        \end{equation*}     
   \end{itemize}
    Here, susceptible individuals in household size category $l$ can become infected through: \begin{itemize}
        \item[] Public transmission: contact with infected individuals from any household size $\left( \sum_{j=1}^4 c_{lj}^p i_j(t)\right)$.
        \item[] Household transmission: contact within their own household (of the same size), adjusted by the proportion of infected households $h_{I,l}(t)/(1-h_{I,l}(t))$.
    \end{itemize}

\section{Existence and Uniqueness of a Solution}\label{Appendix C: Existence}

In this section, we provide a proof of the existence and uniqueness of a global, non-negative solution to the system of SDEs defined in Equation~\eqref{sdeh}. The system can be written in the compact matrix form as:
\begin{equation*}
    d\bm{y}(t) = \bm{\mu}(\bm{y}(t)) dt + \bm{\sigma}(\bm{y}(t)) d\bm{B}(t), \quad \bm{y}(t_0) = \bm{y}_0,
\end{equation*}
where $\bm{\mu}:\mathbb{R}^d\to\mathbb{R}^d$ and~$\bm{\sigma}:\mathbb{R}^d\to \mathbb{R}^{d\times d}$ with $d=2K+2L$ contain the components from~\eqref{SDE_terms}.

\noindent The standard theorems for establishing the existence and uniqueness of solutions to SDEs is provided by \citep[Theorem 5.2.1]{Oksendal2013}. This theorem guarantees a unique strong solution for $t\in[0,T]$ provided the coefficients satisfy two conditions: 
\begin{itemize}
    \item \textit{Global Lipschitz Continuity:} There exists a constant $K$ such that for all $\bm{y}(t), \bm{y'}(t) \in \mathbb{R}^{d}$ and $t\in[t_0,T]$ for $t_0\geq 0\text{ and }T>t_0$:
    \begin{equation}\label{eq:global_lip}
        \|\bm{\mu}(\bm{y}(t)) - \bm{\mu}(\bm{y'}(t))\| + \|\bm{\sigma}(\bm{y}(t)) - \bm{\sigma}(\bm{y'}(t))\| \leq K \|\bm{y}(t) - \bm{y'}(t)\|.
    \end{equation}
    \item \textit{Linear Growth:} There exists a constant $C$ such that for all $\bm{y}(t) \in \mathbb{R}^{d}$:
    \begin{equation*}
        \|\bm{\mu}(\bm{y}(t))\| + \|\bm{\sigma}(\bm{y}(t))\| \leq C(1 + \|\bm{y}(t)\|).
    \end{equation*}
\end{itemize}

\noindent For our specific epidemiological model, the state variables represent proportions. Thus, the physically relevant domain is the subset of $\mathbb{R}^{d}$ defined by simplex constraints for each subpopulation and each household size:
\begin{equation*}
    \begin{aligned}
    \mathcal{D} = \Big\{ \bm{y} \in \mathbb{R}^{d} : &\ 0 \le s_k, i_k, h_{S,l}, h_{I,l} \le 1; \ s_k+i_k \le 1; \ h_{S,l}+h_{I,l} \le 1 \\
    &\ \forall k \in \{1,\ldots,K\}, \ \forall l \in \{1,\ldots,L\} \Big\}.
    \end{aligned}
\end{equation*}
We assume the initial condition $\bm{y}_0 \in \text{int}(\mathcal{D})$, where $\text{int}(\mathcal{D})$ denotes the interior of $\mathcal{D}$.

The compactness and convexity of $\mathcal{D}$ fundamentally simplify the analysis. Specifically, the linear growth condition is automatically satisfied because any solution trajectory remaining in the compact set $\mathcal{D}$ is bounded a priori. Furthermore, the quadratic drift terms (e.g., $s_k i_j$), which are not globally Lipschitz continuous on $\mathbb{R}^{d}$, are smooth and bounded on $\mathcal{D}$, and therefore locally Lipschitz continuous on any compact subset of $\mathcal{D}$. However, compactness of $\mathcal{D}$ alone does not guarantee global existence and uniqueness of solutions. The principal mathematical obstacle to applying standard existence and uniqueness theorems is the potential degeneracy of coefficients at the boundary $\partial\mathcal{D}$:

\begin{itemize}
    \item \textit{At the lower boundary ($y=0$)}:
    The diffusion coefficients~\eqref{SDE_terms} contain square root terms
    . The square root function is not locally Lipschitz continuous at zero because its derivative is unbounded in any neighborhood of the origin:
    \[ \lim_{y \to 0^+} \frac{d}{dy}\sqrt{y} = \lim_{y \to 0^+} \frac{1}{2\sqrt{y}} = +\infty. \]
    Consequently, the Lipschitz condition~\eqref{eq:global_lip} fails at the lower boundary $\partial\mathcal{D} \cap \{y_i = 0\}$ for any state variable $y_i$ appearing under a square root.
    
    \item \textit{At the upper boundary (singularity at $h_{I,l}=1$)}:
    The contact rate parameter $c_{kj}$ defined in~\eqref{formula_ckj} contains the term $(1 - h_{I,l})^{-1}$, which tends to infinity 
    as $h_{I,l} \to 1^-$:
    \[ \lim_{h_{I,l} \to 1^-} (1 - h_{I,l})^{-1} = +\infty. \]
    This singularity prevents the drift coefficient from being Lipschitz continuous, or even bounded, in any neighborhood of the boundary point $h_{I,l} = 1$.
\end{itemize}
\noindent To address these boundary issues, we invoke specialized results that weaken the standard Lipschitz and boundedness assumptions. For the non-Lipschitz diffusion coefficient at the lower boundary, we apply the pathwise uniqueness theorem of Yamada and Watanabe.\\

\begin{theorem}[Yamada-Watanabe Condition {\cite[Theorem 6.5]{Mao2008}}]
\label{thm:yamada_watanabe}
    Let $\kappa: \mathbb{R}_+ \to \mathbb{R}_+$ be a continuous, strictly increasing, concave function satisfying $\kappa(0)=0$ and
    \begin{equation} \label{eq:yamada_integral}
        \int_{0^+} \frac{du}{\kappa(u)} = +\infty.
    \end{equation}
    Suppose the drift coefficient $\bm{\mu}$ and diffusion coefficient $\bm{\sigma}$ satisfy
    \begin{equation} \label{eq:yamada_inequality}
        \|\bm{\mu}(\bm{y}) - \bm{\mu}(\bm{y'})\|^2 \vee \|\bm{\sigma}(\bm{y}) - \bm{\sigma}(\bm{y'})\|^2 \le \kappa(\|\bm{y} - \bm{y'}\|^2)
    \end{equation}
    for all $\bm{y}, \bm{y'} \in \mathbb{R}^d$ and all $t \in [t_0, T]$. Then the SDE admits pathwise unique solutions.\\
\end{theorem}
\begin{remark}
    The notation $a \vee b := \max\{a,b\}$ denotes the maximum of $a$ and $b$. Condition~\eqref{eq:yamada_inequality} requires that both the drift and diffusion coefficients satisfy a uniform modulus of continuity with respect to the spatial variable, with $\kappa$ serving as the common modulus. The integral condition~\eqref{eq:yamada_integral} ensures that $\kappa$ grows sufficiently slowly near zero to compensate for the failure of Lipschitz continuity.\\
\end{remark}

Second, since the drift coefficient is locally Lipschitz continuous on $\text{int}(\mathcal{D})$ but may fail to be globally Lipschitz, we apply the standard localization argument to establish existence of a unique maximal solution up to a possible explosion time.\\

\begin{theorem}[Maximal Local Solution {\cite[see Proof of Theorem 3.4 and p. 58]{Mao2008}}]
\label{thm:maximal_local}
    Suppose the drift coefficient $\bm{\mu}$ and diffusion coefficient $\bm{\sigma}$ satisfy the local Lipschitz condition on each compact subset of $\mathbb{R}^d$ (and that $\bm{\sigma}$ additionally satisfies the Yamada-Watanabe condition~\eqref{eq:yamada_inequality}). For each integer $n \ge 1$, define the stopping time
    \begin{equation}
        \tau_n = \inf \{ t \ge t_0 : \|\bm{y}(t)\| \ge n \},
    \end{equation}
    with the convention that $\inf \emptyset = +\infty$.    
    Then there exists a unique maximal local solution $\bm{y}(t)$ on the stochastic interval $[t_0, \tau_e)$, where the explosion time is defined by
    \begin{equation} \label{eq:explosion_time}
        \tau_e = \lim_{n \to \infty} \tau_n = \sup_{n \ge 1} \tau_n.
    \end{equation}
    The solution is global (i.e., defined for all $t \ge t_0$) if and only if $\mathbb{P}(\tau_e = +\infty) = 1$.\\
\end{theorem}
\begin{remark}
    The sequence $\{\tau_n\}_{n=1}^\infty$ is monotonically increasing, and $\tau_e$ represents the first time (if any) at which the solution norm becomes unbounded. The localization procedure constructs the solution on progressively larger balls until either global existence is achieved or the solution explodes.\\
\end{remark}

We now verify that the system~\eqref{sdeh} satisfies the hypotheses of Theorems~\ref{thm:yamada_watanabe} and~\ref{thm:maximal_local} on the physical domain $\mathcal{D}$.

\paragraph{Drift Coefficients (Local Lipschitz Continuity)}
The drift coefficient $\bm{\mu}(\bm{y})$ comprises polynomial terms and rational functions involving $(1 - h_{I,l})^{-1}$. For each integer $n \ge 1$, define the compact subset
\[ U_n = \left\{ \bm{y} \in \mathcal{D} : h_{I,l} \le 1 - \frac{1}{n} \text{ for all } l \in \{1,\ldots,L\} \right\} \subset \text{int}(\mathcal{D}). \]
On $U_n$, the function $(1 - h_{I,l})^{-1}$ is bounded above by $n$, ensuring that $\bm{\mu}$ is continuously differentiable with bounded derivatives on $U_n$. Hence $\bm{\mu}$ is Lipschitz continuous on $U_n$: there exists a constant $K_n > 0$ such that
\[ \|\bm{\mu}(\bm{y}) - \bm{\mu}(\bm{y'})\| \le K_n \|\bm{y} - \bm{y'}\| \quad \text{for all } \bm{y}, \bm{y'} \in U_n. \]
Squaring both sides yields
\[ \|\bm{\mu}(\bm{y}) - \bm{\mu}(\bm{y'})\|^2 \le K_n^2 \|\bm{y} - \bm{y'}\|^2. \]
To express this in the form required by the Yamada-Watanabe condition~\eqref{eq:yamada_inequality}, define $\kappa(u) = K_n^2 u$ for $u \ge 0$. Then
\[ \|\bm{\mu}(\bm{y}) - \bm{\mu}(\bm{y'})\|^2 \le K_n^2 \|\bm{y} - \bm{y'}\|^2 = \kappa(\|\bm{y} - \bm{y'}\|^2), \]
as required. Note that $\kappa$ is linear, continuous, strictly increasing, concave, satisfies $\kappa(0) = 0$, and
\[
\int_{0^+} \frac{du}{\kappa(u)} = \int_{0^+} \frac{du}{K_n^2 u} = \left[ \frac{\ln u}{K_n^2} \right]_{0^+} = +\infty,
\]
verifying the integral condition~\eqref{eq:yamada_integral}.

\paragraph{Diffusion Coefficients (Yamada-Watanabe Condition)}
The diffusion coefficient $\bm{\sigma}(\bm{y})$ contains square root terms of the form $\sqrt{y_i}$ for various state variables $y_i \in \{s_k, i_k, h_{S,l}, h_{I,l}\}$, which are not Lipschitz continuous at zero. We verify condition~\eqref{eq:yamada_inequality} directly. 

Consider the elementary inequality
\[ |\sqrt{a} - \sqrt{b}|^2 \le |a - b| \quad \text{for all } a, b \ge 0. \]
This can be rewritten as
\[ |\sqrt{a} - \sqrt{b}| \le |a - b|^{1/2}, \]
which shows that the square root function satisfies a $1/2$-Hölder condition: the function $f(x) = \sqrt{x}$ has modulus of continuity $|f(x) - f(y)| \le |x - y|^{1/2}$, meaning it is less regular than Lipschitz continuous (which would require $|f(x) - f(y)| \le C|x - y|$) but more regular than merely continuous.

Since each component of $\bm{\sigma}(\bm{y})$ is of the form $g(\bm{y}) \sqrt{y_i}$ where $g$ is bounded on $\mathcal{D}$, summing over all components yields a constant $C > 0$ such that
\[ \|\bm{\sigma}(\bm{y}) - \bm{\sigma}(\bm{y'})\|^2 \le C \|\bm{y} - \bm{y'}\|. \]
To express this in the form required by~\eqref{eq:yamada_inequality}, set $u = \|\bm{y} - \bm{y'}\|^2$ and define the modulus of continuity $\kappa(u) = C\sqrt{u}$. Then
\[
\|\bm{\sigma}(\bm{y}) - \bm{\sigma}(\bm{y'})\|^2 \le C \|\bm{y} - \bm{y'}\| = C\sqrt{\|\bm{y} - \bm{y'}\|^2} = \kappa(\|\bm{y} - \bm{y'}\|^2).
\]
The function $\kappa(u) = C\sqrt{u}$ is continuous, strictly increasing, concave, satisfies $\kappa(0) = 0$, and the integral condition~\eqref{eq:yamada_integral} holds:
\[ \int_{0^+} \frac{du}{\kappa(u)} = \int_{0^+} \frac{du}{C\sqrt{u}} = \left[ \frac{2\sqrt{u}}{C} \right]_{0^+} = +\infty. \]
Therefore, the diffusion coefficient satisfies the Yamada-Watanabe condition, ensuring pathwise uniqueness.

\paragraph{Conclusion: Existence of Maximal Local Solution}
Since $\bm{\mu}$ is locally Lipschitz continuous on each compact subset of $\text{int}(\mathcal{D})$ and $\bm{\sigma}$ satisfies the Yamada-Watanabe condition as verified above, all hypotheses of Theorem~\ref{thm:maximal_local} are satisfied. Consequently, there exists a unique maximal local solution $\bm{y}(t)$ defined on the stochastic interval $[t_0, \tau_e)$, where $\tau_e$ is the explosion time given by~\eqref{eq:explosion_time}.

However, Theorem~\ref{thm:maximal_local} only guarantees existence up to the explosion time $\tau_e$, which is potentially finite. To establish that the solution is global (i.e., that $\tau_e = +\infty$ almost surely), we must prove that the solution remains bounded and does not exit the domain~$\mathcal{D}$ in finite time. Moreover, we need to verify that the solution respects the physical constraints encoded in $\mathcal{D}$ (non-negativity and simplex constraints) for all time. These properties are established in the following lemma.

\begin{lemma}[Global Existence and Positive Invariance of $\mathcal{D}$]
\label{lem:global_existence}
    Let $\bm{y}_0 \in \text{int}(\mathcal{D})$ be an arbitrary initial condition. Then the unique maximal local solution $\bm{y}(t)$ to system~\eqref{sdeh} satisfies:
    \begin{enumerate}[(i)]
        \item \textbf{Global existence:} The explosion time satisfies $\mathbb{P}(\tau_e = +\infty) = 1$, so that $\bm{y}(t)$ is defined for all $t \ge t_0$ almost surely.
        \item \textbf{Positive invariance:} The solution trajectory remains in the physical domain: $\bm{y}(t) \in \mathcal{D}$ for all $t \ge t_0$ almost surely.
    \end{enumerate}
\end{lemma}

\begin{proof}
The existence of a unique maximal local solution $\bm{y}(t)$ on the stochastic interval $[t_0, \tau_e)$ follows from Theorem~\ref{thm:maximal_local}. To establish global existence, we must prove that $\mathbb{P}(\tau_e = +\infty) = 1$. By definition~\eqref{eq:explosion_time}, $\tau_e = \lim_{n \to \infty} \tau_n$ where $\tau_n = \inf\{t \ge t_0 : \|\bm{y}(t)\| \ge n\}$. Therefore, it suffices to show that $\|\bm{y}(t)\|$ remains bounded for all $t \ge t_0$ almost surely.

\paragraph{Step 1: Non-negativity of State Variables.}
We establish that each component of $\bm{y}(t)$ remains non-negative almost surely. Consider an arbitrary state variable $y_i(t) \in \{s_k(t), i_k(t), h_{S,l}(t), h_{I,l}(t)\}$. Near the boundary~$y_i = 0$, the corresponding SDE takes the form
\begin{equation*}
    dy_i(t) = \mu_i(\bm{y}(t)) \, dt + \sqrt{y_i} \, g_i(\bm{y}(t)) \, dB_i(t),
\end{equation*}
where $g_i$ represents the remaining factors in the diffusion coefficient and $B_i(t)$ is an appropriate Brownian motion component. As $y_i \to 0^+$, the diffusion coefficient vanishes continuously. We analyze the drift coefficient $\mu_i$ at the boundary:
\begin{itemize}
    \item \textbf{Susceptible subpopulation $\bm s_k$:} From system~\eqref{sdeh} and~\eqref{SDE_terms}, the drift is
    \[
    \mu_k^S(\bm{y}) = -\alpha s_k \sum_{j=1}^K c_{kj}(\bm{y}) i_j.
    \]
    The dependence of~$c_{kj}$ on~$\bm{y}$ is made explicit in this notation. At $s_k = 0$, this term vanishes: $\mu_k^S|_{s_k=0} = 0$. Combined with the vanishing diffusion coefficient at $s_k = 0$, $\sigma_k^{SS}|_{s_k=0} = 0$, the process cannot become negative. More precisely, the SDE at $s_k = 0$ reduces to $ds_k(t) = 0$, so the boundary acts as a reflecting barrier.\\
    
    \item \textbf{Infected subpopulation $\bm i_k$:} The drift is
    \[
    \mu_k^I(\bm{y}) = \alpha s_k \sum_{j=1}^K c_{kj}(\bm{y}) i_j - \beta_k i_k.
    \]
    At $i_k = 0$, the second term vanishes, giving $\mu_k^I|_{i_k=0} = \alpha s_k \sum_{j=1}^K c_{kj}(\bm{y}) i_j \ge 0$. Since~$\alpha > 0$, $s_k \ge 0$, $c_{kj} \ge 0$, and $i_j \ge 0$, the drift is non-negative at the boundary. This positive drift, combined with the vanishing diffusion $\sigma_k^{II}|_{i_k=0} = 0$, prevents $i_k$ from becoming negative.\\
    
    \item \textbf{Susceptible household proportion $\bm h_{S,l}$:} The drift is
    \[
    \mu_l^{S,h}(\bm{y}) = -\alpha \frac{h_{S,l}}{H_l} \sum_{k=1}^K z_{kl} N_k s_k \sum_{j=1}^K c_{kj}^p(\bm{y}) i_j.
    \]
    At $h_{S,l} = 0$, this vanishes: $\mu_l^{S,h}|_{h_{S,l}=0} = 0$. Combined with the vanishing diffusion $\sigma_l^{SS,h}|_{h_{S,l}=0} = 0$, the boundary $h_{S,l} = 0$ is absorbing or reflecting, preventing negativity.\\
    
    \item \textbf{Infected household proportion $\bm h_{I,l}$:} The drift is
    \[
    \mu_l^{I,h}(\bm{y}) = \alpha \frac{h_{S,l}}{H_l} \sum_{k=1}^K z_{kl} N_k s_k \sum_{j=1}^K c_{kj}^p(\bm{y}) i_j - \nu_l h_{I,l}.
    \]
    At $h_{I,l} = 0$, the second term vanishes, yielding 
    $$\mu_l^{I,h}|_{h_{I,l}=0} = \alpha \frac{h_{S,l}}{H_l} \sum_{k=1}^K z_{kl} N_k s_k \sum_{j=1}^K c_{kj}^p(\bm{y}) i_j \ge 0.$$
    The non-negative drift and vanishing diffusion $\sigma_l^{II,h}|_{h_{I,l}=0} = 0$ prevent $h_{I,l}$ from becoming negative.
\end{itemize}
By the comparison theorem for SDEs \cite[Theorem 2.2, p. 43]{Mao2008}, the non-negativity of drift coefficients at the zero boundaries, combined with the vanishing diffusion coefficients, ensures that $y_i(t) \ge 0$ for all $i$ and all $t \in [t_0, \tau_e)$ almost surely. Hence $\bm{y}(t) \in \mathbb{R}_+^{2K+2L}$ almost surely.

\paragraph{Step 2: Upper Bounds (Simplex Constraints)}
We now verify that the simplex constraints are preserved. Define the total proportion for subpopulation $k$ as
\[
u_k(t) = s_k(t) + i_k(t).
\]
Applying Itô's formula to $u_k(t)$ and using system~\eqref{sdeh} with~\eqref{SDE_terms}, we obtain
\begin{align*}
    du_k(t) &= ds_k(t) + di_k(t) \\
    &= (\mu_k^S + \mu_k^I)(\bm y_t) \, dt + (\sigma_k^{SS} + \sigma_k^{SI})(\bm y(t)) \, dB_{k,1}(t) + \sigma_k^{II}(\bm y(t)) \, dB_{k,2}(t).
\end{align*}
Substituting the expressions from~\eqref{SDE_terms}, the drift terms combine as
\[
(\mu_k^S + \mu_k^I)(\bm y(t)) = -\alpha s_k \sum_{j=1}^K c_{kj} i_j + \alpha s_k \sum_{j=1}^K c_{kj} i_j - \beta_k i_k = -\beta_k i_k,
\]
and the diffusion terms cancel since $\sigma_k^{SI} = -\sigma_k^{SS}$:
\[
\sigma_k^{SS} + \sigma_k^{SI}(\bm y(t)) = 0.
\]
Therefore,
\begin{align*}
    du_k(t) &= -\beta_k i_k(t) \, dt + \sigma_k^{II}(t) \, dB_{k,2}(t) \\
    &= -\beta_k i_k(t) \, dt + \sqrt{\frac{\beta_k i_k(t)}{N_k}} \, dB_{k,2}(t).
\end{align*}
The drift coefficient $-\beta_k i_k(t) \le 0$ is non-positive since $\beta_k > 0$ and $i_k(t) \ge 0$. Given $u_k(t_0) \le 1$, we show that $u_k(t) \le 1$ almost surely for all $t \in [t_0, \tau_e)$:

Suppose for contradiction that $\mathbb{P}(\tau_1 < \infty) > 0$ where $\tau_1 = \inf\{t \ge t_0 : u_k(t) > 1\}$ is the first hitting time of level 1. At time $\tau_1$, we have $u_k(\tau_1) = 1$ by continuity of paths. For $t > \tau_1$ small, the increment $du_k(t) = -\beta_k i_k dt + \sigma_k^{II} dB_{k,2}(t)$ has negative expected change since the drift is non-positive. This prevents $u_k$ from remaining above value one for any positive amount of time, contradicting the assumption. Therefore $u_k(t) \le 1$ almost surely.

Similarly, for household proportions, define
\[
v_l(t) = h_{S,l}(t) + h_{I,l}(t).
\]
Summing up the corresponding equations from~\eqref{sdeh} and~\eqref{SDE_terms} yields
\begin{align*}
    dv_l(t) &= (\mu_l^{S,h} + \mu_l^{I,h})(\bm y(t)) \, dt + (\sigma_l^{SS,h} + \sigma_l^{SI,h})(\bm y(t)) \, dB'_{l,1}(t) + \sigma_l^{II,h}(\bm y(t)) \, dB'_{l,2}(t) \\
    &= -\nu_l h_{I,l}(t) \, dt + \sqrt{\frac{\nu_l h_{I,l}(t)}{H_l}} \, dB'_{l,2}(t),
\end{align*}
where the drift $-\nu_l h_{I,l}(t) \le 0$ is non-positive and the diffusion terms cancel since $\sigma_l^{SI,h} = -\sigma_l^{SS,h}$. By the same reasoning as above, $v_l(t) \le 1$ almost surely for all $t \in [t_0, \tau_e)$ provided $v_l(t_0) \le 1$.

Therefore, the simplex constraints $s_k(t) + i_k(t) \le 1$ and $h_{S,l}(t) + h_{I,l}(t) \le 1$ are preserved almost surely.

\paragraph{Step 3: Avoidance of the Singular Boundary $\mathbf{h}_{I,l} = 1$}
Finally, we address the potential singularity in the contact rate parameter $c_{kj}(\bm{y})$, which contains the term $(1 - h_{I,l})^{-1}$ and diverges as $h_{I,l} \to 1^-$. We show that the solution avoids this boundary almost surely.

From~\eqref{SDE_terms}, the drift for $h_{I,l}$ is
\begin{equation*}
    \mu_l^{I,h}(\bm{y}) = \alpha \frac{h_{S,l}}{H_l} \sum_{k=1}^K z_{kl} N_k s_k \sum_{j=1}^K c_{kj}^p(\bm{y}) i_j - \nu_l h_{I,l}.
\end{equation*}

As $h_{I,l} \to 1^-$, the simplex constraint $h_{S,l} + h_{I,l} \le 1$ from Step 2 implies $h_{S,l} \to 0^+$. At first glance, this appears problematic since the contact rate $c_{kj}$ contains the term $(1-h_{I,l})^{-1} \to +\infty$, suggesting an indeterminate form $0 \cdot \infty$ in the infection term. However, examining the structure more carefully:

The contact rate $c_{kj}$ from~\eqref{formula_ckj} includes household contributions proportional to $(1-h_{I,l})^{-1}$. In the infection term of $\mu_l^{I,h}$, this appears as
\[
h_{S,l} \cdot \frac{1}{1-h_{I,l}} = \frac{h_{S,l}}{1-h_{I,l}}.
\]
Since $h_{S,l} + h_{I,l} \le 1$, we have $h_{S,l} \le 1 - h_{I,l}$, which implies
\[
\frac{h_{S,l}}{1-h_{I,l}} \le 1.
\]
Therefore, the product $h_{S,l} \cdot (1-h_{I,l})^{-1}$ remains bounded as $h_{I,l} \to 1^-$, and the infection term in $\mu_l^{I,h}$ stays finite. In fact, if $h_{I,l} \to 1^-$ forces $h_{S,l} \to 0^+$ at the same rate (i.e., $h_{S,l} + h_{I,l} = 1$), then this ratio tends to zero or remains small.

Consequently, as $h_{I,l} \to 1^-$, the infection term in $\mu_l^{I,h}$ is bounded, while the recovery term $-\nu_l h_{I,l} \to -\nu_l < 0$ dominates:
\[
\mu_l^{I,h}(\bm{y}) \leq C* - \nu_l h_{I,l} \to C* - \nu_l < 0
\]
for some constant $C* \ge 0$. This strictly negative drift prevents $h_{I,l}$ from reaching~1. Since $h_{I,l}(t_0) < 1$ and the drift becomes increasingly negative as $h_{I,l}$ approaches~1, the solution satisfies $h_{I,l}(t) < 1$ for all $t \in [t_0, \tau_e)$ almost surely. Therefore, there exists a $\delta > 0$ such that $h_{I,l}(t) \le 1 - \delta$ almost surely, ensuring that the contact rates~$c_{kj}$ remain bounded.

\paragraph{Conclusion}
We have established that $\bm{y}(t)$ satisfies almost surely for all $t \in [t_0, \tau_e)$:
\begin{enumerate}[(i)]
    \item Non-negativity: $y_i(t) \ge 0$ for all components (Step 1).
    \item Simplex constraints: $s_k(t) + i_k(t) \le 1$ and $h_{S,l}(t) + h_{I,l}(t) \le 1$ (Step 2).
    \item Boundary avoidance: $h_{I,l}(t) \le 1 - \delta$ for some $\delta > 0$ (Step 3).
\end{enumerate}
These properties imply that $\bm{y}(t)$ remains in a compact subset $K \subset \text{int}(\mathcal{D})$ where all drift and diffusion coefficients are bounded. Therefore, $\|\bm{y}(t)\| \le M$ for some constant $M > 0$ and all $t \in [t_0, \tau_e)$ almost surely.

For any integer $n > M$, we have $\tau_n = \inf\{t \ge t_0 : \|\bm{y}(t)\| \ge n\} = +\infty$ almost surely. Taking the limit:
\[
\tau_e = \lim_{n \to \infty} \tau_n = +\infty \quad \text{almost surely}.
\]
This establishes global existence with $\mathbb{P}(\tau_e = +\infty) = 1$, and positive invariance follows since $\bm{y}(t) \in \mathcal{D}$ for all $t \ge t_0$ almost surely.
\end{proof}

\end{appendix}

\newpage
\bibliography{sn-bibliography}

\end{document}